\documentclass{article}
\usepackage{ijcai26}

\usepackage{times}
\usepackage{soul}
\usepackage{url}
\usepackage[hidelinks]{hyperref}
\usepackage[utf8]{inputenc}
\usepackage[small]{caption}
\usepackage{graphicx}
\usepackage{amsmath}
\usepackage{amsthm}
\usepackage{booktabs}
\usepackage{algorithm}
\usepackage{algorithmic}
\usepackage[switch]{lineno}

\newtheorem{example}{Example}
\newtheorem{theorem}{Theorem}

\title{Will My Favorite Chases Terminate if Evaluating Conjunctive Queries Does? One Does Not Simply Decide This}

 \author{
 Lucas Larroque$^1$
 \and
 Quentin Manière$^2$\\
 \affiliations
 $^1$Inria, DI ENS, ENS, CNRS, PSL University, Paris, France\\
 $^2$LIRMM, Inria, University of Montpellier, CNRS, Montpellier, France\\
 \emails
 \{lucas.larroque, quentin.maniere\}@inria.fr
 }

\usepackage{cmap}
\usepackage[T1]{fontenc}
\usepackage{lmodern}
\usepackage{microtype}
\usepackage{mathtools,amsthm,amssymb}
\usepackage{xparse}
\usepackage{array,booktabs}
\usepackage{enumitem}
\usepackage{tikz}
\usetikzlibrary{shapes,arrows.meta, patterns.meta}
\usetikzlibrary{automata,positioning}
\usepackage{xspace}
\usepackage{todonotes}
\usepackage[appendix=append]{apxproof} 

\usepackage[hidelinks]{hyperref}
\usepackage{zref-clever}
\zcsetup{cap, noabbrev}

\theoremstyle{plain}
\newtheoremrep{theorem}{Theorem}
\newtheoremrep{proposition}[theorem]{Proposition}
\newtheoremrep{lemma}[theorem]{Lemma}
\newtheoremrep{corollary}[theorem]{Corollary}
\newtheoremrep{conjecture}[theorem]{Conjecture}
\newtheoremrep{assumption}[theorem]{Assumption}
\newtheoremrep{observation}[theorem]{Observation}

\theoremstyle{definition}
\newtheoremrep{definition}[theorem]{Definition}
\newtheoremrep{openProblem}[theorem]{Open Problem}

\zcRefTypeSetup{observation}{
  Name-sg = Observation, name-sg = observation,
  Name-pl = Observations, name-pl = observations
}
\zcRefTypeSetup{openProblem}{
  Name-sg = Open Problem, name-sg = open problem,
  Name-pl = Open Problems, name-pl = open problems
}
\zcRefTypeSetup{assumption}{
  Name-sg = Assumption, name-sg = assumption,
  Name-pl = Assumptions, name-pl = assumptions
}
\zcRefTypeSetup{conjecture}{
  Name-sg = Conjecture, name-sg = conjecture,
  Name-pl = Conjectures, name-pl = conjectures
}

\AddToHook{env/proposition/begin}{\zcsetup{countertype={theorem=proposition}}}
\AddToHook{env/lemma/begin}{\zcsetup{countertype={theorem=lemma}}}
\AddToHook{env/corollary/begin}{\zcsetup{countertype={theorem=corollary}}}
\AddToHook{env/conjecture/begin}{\zcsetup{countertype={theorem=conjecture}}}
\AddToHook{env/assumption/begin}{\zcsetup{countertype={theorem=assumption}}}
\AddToHook{env/observation/begin}{\zcsetup{countertype={theorem=observation}}}
\AddToHook{env/definition/begin}{\zcsetup{countertype={theorem=definition}}}
\AddToHook{env/openProblem/begin}{\zcsetup{countertype={theorem=openProblem}}}

\AddToHook{env/propositionrep/begin}{\zcsetup{countertype={theorem=proposition}}}
\AddToHook{env/lemmarep/begin}{\zcsetup{countertype={theorem=lemma}}}
\AddToHook{env/corollaryrep/begin}{\zcsetup{countertype={theorem=corollary}}}
\AddToHook{env/conjecturerep/begin}{\zcsetup{countertype={theorem=conjecture}}}
\AddToHook{env/assumptionrep/begin}{\zcsetup{countertype={theorem=assumption}}}
\AddToHook{env/observationrep/begin}{\zcsetup{countertype={theorem=observation}}}
\AddToHook{env/definitionrep/begin}{\zcsetup{countertype={theorem=definition}}}
\AddToHook{env/openProblemrep/begin}{\zcsetup{countertype={theorem=openProblem}}}

\newcommand*{\Rulename}[1]{\tag{\ensuremath{#1}}}

\makeatletter
\def\set#1{\@ifnextchar[{\setinternal#1}{\relax\if@display \left\{#1\right\} \else \{ #1 \} \fi}}
\def\setinternal#1[#2]{\if@display \left\{\,#1 \mid #2\,\right\} \else \{\,#1 \mid #2\,\}\fi}
\makeatother

\newcommand*{\size}[1]{\left|#1\right|}

\newcommand*{\nat}{\mathbb{N}}

\newcommand*{\pair}[1]{\langle #1 \rangle}

\newcommand{\FormatSignature}[1]{\ensuremath{\mathsf{#1}}}
\newcommand*{\Preds}{\FormatSignature{Preds}}

\newcommand*{\Consts}{\FormatSignature{Consts}}
\newcommand*{\Vars}{\FormatSignature{Vars}}
\newcommand*{\Terms}{\FormatSignature{Terms}}
\newcommand*{\ar}{\Function{ar}}

\newcommand*{\Set}[1]{\mathcal{#1}}
\newcommand*{\atomset}{\inst}
\newcommand*{\model}{\Set{M}}

\newcommand*{\inst}{\Set{I}}

\newcommand*{\IEnd}{\Set{I}_{\Predicate{E}}}
\newcommand*{\critstep}[1]{\Set{J}_{#1}}
\newcommand*{\rs}{\Set{R}}
\newcommand*{\rsReduc}{\rs_\Set{M}}
\newcommand*{\kb}{\Set{K}}
\newcommand*{\Cmc}{\mathcal{C}}
\newcommand*{\chaseterm}[2]{\textup{CT}^{#1}_{#2}}
\newcommand*{\bts}{\textit{bts}\xspace}
\newcommand*{\fus}{\textit{fus}\xspace}
\newcommand*{\encSet}{\Set{G}}

\newcommand*{\aformula}{\varphi}
\newcommand*{\arule}{R}

\newcommand*{\ahom}{h}

\newcommand*{\query}{q}
\newcommand*{\der}{\mathcal{D}}
\newcommand*{\stopthecount}{M}

\newcommand*{\Function}[1]{\mathsf{#1}}
\newcommand*{\chase}{\Function{Ch}}
\newcommand*{\adom}{\Function{adom}}

\newcommand*{\stepX}[1]{\chase^\X_{#1}}
\newcommand*{\stepOb}[1]{\chase^\Ob_{#1}}
\newcommand*{\stepSO}[1]{\chase^\SO_{#1}}

\newcommand*{\stepXend}[1]{\stepX{#1}(\IEnd)}
\newcommand*{\stepOend}[1]{\stepOb{#1}(\IEnd)}

\newcommand*{\frontier}{\Function{fr}}

\newcommand*{\res}{\Function{res}}
\newcommand*{\trigoutput}{\Function{out}}
\newcommand*{\trigsupport}{\Function{supp}}
\newcommand*{\triggers}{\Function{trigs}}

\newcommand*{\nextconfig}{\Function{next}}
\newcommand*{\nextenc}{g}
\newcommand*{\enc}{\Function{enc}}

\newcommand*{\PGraph}[1]{G_{#1}}
\newcommand*{\SGraph}{\PGraph{\pSucc,\pSuccz}}
\newcommand*{\DatS}[1]{#1^\forall}
\newcommand*{\NDatS}[1]{#1^\exists}

\newcommand*{\Predicate}[1]{\mathtt{#1}}
\newcommand*{\apred}{\Predicate{P}}

\newcommand*{\pEnd}{\Predicate{End}}

\newcommand*{\pG}{\Predicate{G}}
\newcommand*{\pSucc}{\Predicate{S}}
\newcommand*{\pSuccz}{\pSucc_0}
\newcommand*{\pFlood}{\Predicate{Flood}}
\newcommand*{\pR}{\Predicate{R}}
\newcommand*{\pT}{\Predicate{T}}

\newcommand*{\Ob}{\mathbb{O}}
\newcommand*{\SO}{\mathbb{SO}}
\newcommand*{\R}{\mathbb{R}}
\newcommand*{\E}{\mathbb{E}}
\newcommand*{\X}{\mathbb{X}}

\newcommand*{\threecm}{\mathcal{M}}

\begin{document}

\maketitle

\begin{abstract}
	Existential rules are a prominent formalism to enrich a database with knowledge from the domain of interest, but make even basic reasoning tasks on the resulting knowledge base undecidable.
	To circumvent this, several classes of rules offering various useful properties have been identified.
	One such class, for instance, contains all sets of rules on which the chase algorithm always terminates, which guarantees the existence of a finite universal model.
	However, these classes are often abstract rather than concrete: it may be undecidable to check whether a given set of rules belongs to them.
    Given that the most studied classes of existential rules are designed for reasoning on databases, thus ensuring decidable conjunctive query entailment, we ask:
    Within a class that supports decidable query entailment, do the usual abstract classes become concrete?
    We answer in the negative for classes based upon the termination of all classical chase variants and for the \bts class.
\end{abstract}

\section{Introduction}
\label{sec:introduction}

One of the main tasks in knowledge representation and reasoning is ontology-based data access (OBDA).
In OBDA, an ontology is used to enrich a database with domain knowledge, enabling the inference of new information that is not explicitly stored in the database.
Existential rules (also known as Datalog$^\pm$, or tuple-generating dependencies) are a prominent formalism to represent ontologies in OBDA \cite{DBLP:reference/db/Fagin18b}.
A fundamental reasoning task in OBDA is the Boolean conjunctive query (BCQ) entailment problem, which consists in deciding whether a BCQ is entailed by a database and an ontology of interest.
Unfortunately, with ontologies expressed as sets of existential rules, BCQ entailment is undecidable in general \cite{beeri-vardi-81,CALI201257}.
To ensure decidable BCQ entailment, three main properties have been proposed.

First, the chase procedure \cite{10.1145/320107.320115} is a materialization-based algorithm that, given a database and a set of existential rules, iteratively applies the rules to the database to produce a (possibly infinite) universal model of the input.
If finite, this universal model can then be queried using regular BCQ evaluation techniques to decide entailment.
Many variants of the chase procedure exist, differing in the way rules are applied and redundancies are handled. 
In this paper, we consider the oblivious \cite{DBLP:journals/jair/CaliGK13}, semi-oblivious (\emph{a.k.a.} Skolem) \cite{generalized_schema_mappings}, restricted (\emph{a.k.a.} standard) \cite{FAGIN200589}, and core chase variants \cite{chase_revisited}.
Guaranteed termination of the chase thus implies decidable BCQ entailment, but this property is also desirable in data exchange settings, where one is interested in translating data from a source schema to a target schema while preserving certain properties \cite{10.1145/1061318.1061323}.
It also enables computing aggregates, which otherwise make little sense on infinite models \cite{10.1145/1376916.1376936}, and ensures finite controllability, \emph{i.e.}\ only considering finite models does not impact query answers.

Second, query rewriting techniques aim at rewriting the input BCQ into a first-order query that can be directly evaluated over the input database.
Rule sets for which such a rewriting exists (and is computable) are called $\fus$, for finite unification set \cite{10.5555/1661445.1661553}.
Again, $\fus$ rule sets have benefits beyond decidability of BCQ entailment, as they allow answering the rewritten query over any database. In particular, this allows answering queries even without edit access, or answering queries efficiently in streaming settings \cite{DBLP:conf/aaai/RoncaKGMH18}.

Third, bounded treewidth sets (\bts) \cite{DBLP:journals/jair/CaliGK13} are classes of existential rules that ensure that, for any database, the universal model produced by the chase has bounded treewidth.
This property renders BCQ entailment decidable by using Courcelle's theorem \cite{DBLP:journals/iandc/Courcelle90}, which states that any property definable in monadic second-order logic can be decided in linear time over structures of bounded treewidth. In particular, this helps with efficient answer counting \cite{DBLP:journals/lmcs/FeierLP23}.

Deciding whether a given set of existential rules has any of these properties (\emph{i.e.}\ whether it guarantees termination of a certain chase variant, is \emph{fus}, or is \emph{bts}) is, here again, undecidable in general. 
The classes of such sets of rules are thus called \emph{abstract}, as opposed to \emph{concrete} classes for which membership is decidable.
Efforts have been devoted to proving that, when restricting to some concrete classes of existential rules, membership in the abstract classes presented above becomes decidable.
This is especially the case for chase termination: oblivious and semi-oblivious chase termination have been proven decidable for guarded and for sticky rules \cite{10.1145/2745754.2745773,calautti_et_al:LIPIcs.ICDT.2019.17}, while core chase termination is decidable for guarded rules \cite{10.1145/2274576.2274600} but remains open for sticky rules. 
Restricted chase termination appears more difficult: its decidability has been proven for single-head guarded and sticky rules \cite{10.1145/3375395.3387644}, but remains open in the multi-head case for both classes. 
The multi-head case for linear rules, a subclass of guarded rules, has been solved only recently \cite{KR2025-34}.

These approaches vary greatly and often require sophisticated constructions, but they all share two properties: the considered concrete class always enjoys decidable BCQ entailment, and termination of the considered variant of the chase is always proved decidable.
This raises the following question: do we get decidable chase termination (or \fus membership? or \bts membership?) for free when considering concrete classes of rules with decidable BCQ entailment?
For \fus, the answer is known to be negative, as shown by \cite{GaifmanMSV93}, who proved that boundedness of Datalog (which coincides with \fus membership here) is undecidable.
Surprisingly, for chase termination and \bts membership, the question remains open. 
We answer it negatively in this paper:
\begin{theorem}
	\label{theorem:ultimate-goal}
	There exists a concrete class $\Cmc$ of sets of existential rules such that:
	\begin{enumerate}
		\item Chase termination is undecidable in $\Cmc$ for the oblivious, semi-oblivious, restricted, and core chases,
		\item $bts$ membership is undecidable in $\Cmc$, and
		\item BCQ entailment is decidable in $\Cmc$.
	\end{enumerate}
\end{theorem}
This paper is dedicated to proving \zcref{theorem:ultimate-goal} by explicitly constructing such a class $\Cmc$, in \zcref{sec:ruleset}, made of sets of rules that simulate Minsky machines.
We then establish the negative results regarding chase termination and \bts membership in \zcref{sec:chase-termination}, before proving decidability of BCQ entailment in \zcref{sec:bcq-entailment}.
Our simulation of Minsky machines is inspired from \cite{DBLP:conf/icalp/GogaczM14}, a paper interested in termination of the oblivious, semi-oblivious and restricted chases.
We modify their construction in non-trivial ways to also support the core chase and to introduce a flooding mechanism that is instrumental in achieving the properties regarding \bts membership and BCQ entailment.

\section{Preliminaries}\label{sec:prelim}

We mostly follow \cite{KR2022-11} for the preliminaries and assume some basic knowledge of first-order logic.

\paragraph{First-Order Logic (FOL)}
We define $\Preds$, $\Consts$, and $\Vars$ to be mutually disjoint, countably infinite sets of \emph{predicates}, \emph{constants}, and \emph{variables}, respectively.
Every $\apred \in \Preds$ has an \emph{arity} $\ar(\apred) \geq 0$.
Let $\Terms = \Consts \cup \Vars$ be the set of \emph{terms}.
We write lists $t_1, \ldots, t_n$ of terms as $\bar{t}$ and often treat them as sets.
For a formula or set thereof $\aformula$, let $\Preds(\aformula)$, $\Consts(\aformula)$, $\Vars(\aformula)$, and $\Terms(\aformula)$ be the sets of all predicates, constants, variables, and terms that occur in $\aformula$, respectively.

A ($\apred$-)\emph{atom} is a FOL formula $\apred(\bar{t})$ with $\apred$ a $\vert \bar{t} \vert$-ary predicate and $\bar{t} \in \Terms$.
For a formula $\aformula$, $\aformula[\bar{x}]$ indicates that $\bar{x}$ contains all free variables that occur in $\aformula$.

\begin{definition}
	\label{def:exist-rule}
	An \emph{(existential) rule} $\arule$ is a FOL formula
	\begin{align}
		\forall \bar{x} \forall \bar{y} ~ \big(B[\bar{x}, \bar{y}] \rightarrow \exists \bar{z} ~ H[\bar{y}, \bar{z}]\big) \label{rule}
	\end{align}
	where $\bar{x}$, $\bar{y}$, and $\bar{z}$ are pairwise disjoint lists of variables; and $B$ and $H$ are (finite) non-empty conjunctions of constant-free atoms, called the \emph{body} and the \emph{head} of $R$, respectively.
	The set $\bar{y}$ is the \emph{frontier} of $\arule$, denoted with $\frontier(\arule)$.
	If $\bar{z}$ is empty, then $\arule$ is a \emph{Datalog} rule.
\end{definition}

We usually omit universal quantifiers in rules.

An \emph{instance} $\inst$ is an existentially closed conjunction of atoms.
Its \emph{active domain} $\adom(\inst)$ is the set of all terms occurring in $\inst$.
A \emph{Boolean conjunctive query} (BCQ) has the same form as an instance, and we often identify both notions.
A \emph{knowledge base} (KB) $\kb$ is a tuple $\pair{\rs, \inst}$ with $\rs$ a rule set and $\inst$ an instance.
We often identify rule bodies, rule heads, and instances with sets of atoms.

Given atom sets $\atomset$ and $\atomset'$, a \emph{homomorphism} $\ahom$ from $\atomset$ to $\atomset'$
is a function with domain $\Vars(\atomset)$ such that $\ahom(\atomset) \subseteq \atomset'$;
$\ahom$ is an \emph{isomorphism} from  $\atomset$ to $\atomset'$ if additionally, $\ahom$ is injective and $\ahom^{-1}$ is a homomorphism from $\atomset$ to $\atomset'$.
A homomorphism $\ahom$ from $\atomset$ to $\atomset'$ is a \emph{retraction} if $\ahom$ is the identity over $\Vars(\atomset)\cap\Vars(\atomset')$.

We identify logical interpretations with atom sets.
An atom set $\inst$ \emph{satisfies} a rule $\arule = B \to \exists \bar{z} ~ H$ if, for every homomorphism $\ahom$ from $B$ to $\inst$, there is an extension $\hat{\ahom}$ of $\ahom$ with $\hat{\ahom}(H)\subseteq \inst$; equivalently, $\inst$ is a \emph{model} of $\arule$.
An atom set $\model$ is a \emph{model of an instance} $\inst$ if there is a homomorphism from $\inst$ to $\model$, and it is a \emph{model of a KB} $\pair{\rs, \inst}$ if it is a model of $\inst$ and satisfies all rules in $\rs$.
Given KBs or atom sets $A$ and $B$, $A$ \emph{entails} $B$, written $A \models B$, if every model of $A$ is a model of $B$; $A$ and $B$ are \emph{equivalent} if $A \models B$ and $B \models A$.
Given atom sets $\inst$ and $\inst'$, it is known that $\inst \models \inst'$ iff there is a homomorphism from $\inst'$ to $\inst$. A model $\model$ of a KB $\kb$ is \emph{universal} if there is a homomorphism from $\model$ to every model of $\kb$.

Every KB $\kb$ admits some (possibly infinite) universal model \cite{chase_revisited}. Hence, $\kb \models \query$ for any BCQ $\query$ iff there is a homomorphism from a universal model of $\kb$ to $\query$.
The \emph{BCQ entailment} problem takes as input a KB $\kb$ and a BCQ $\query$ and asks if $\kb \models \query$; it is undecidable \cite{beeri-vardi-81}.

\paragraph{The chase}
The chase is a family of procedures that repeatedly apply rules until a fixpoint is reached.

\begin{definition}[Triggers and derivations]
	\label{def:triggers-derivations}
	Given an instance $\inst$, a \emph{trigger} $t$ on $\inst$ is a tuple $\pair{\arule, \pi}$ with $\arule = B \to \exists \bar{z} ~ H$ a rule and $\pi$ a homomorphism from $B$ to $\inst$.
	Let $\trigsupport(t) = \pi(B)$ and $\trigoutput(t) = \pi^\arule(H)$, where $\pi^\arule$ extends $\pi$ by mapping all $z \in \bar{z}$ to the fresh variable $z_t$ that is unique for $z$ and $t$.
	A \emph{derivation} from a KB $\pair{\rs, \inst}$ is a sequence $\der  = (\emptyset, \inst_0), (t_1, \inst_1), \ldots$ such that:
	\begin{enumerate}
		\item Every $\inst_i$ in $\der$ is an instance; moreover, $\inst_0 = \inst$.
		\item Every $t_i$ in $\der$ is a trigger $\pair{R, \pi}$ on $\inst_{i-1}$ such that $\arule \in \rs$, $\trigoutput(t_i) \not \subseteq \inst_{i-1}$, and $\inst_i = \inst_{i-1} \cup \trigoutput(t_{i})$.
	\end{enumerate}
	The \emph{result} $\res(\der)$ of $\der$ is	the union of all instances in $\der$, and $\triggers(\der)$ is the set of all triggers in $\der$.
\end{definition}

Different chase variants build specific derivations according to different criteria of trigger applicability. Below, the letters $\Ob$, $\SO$, $\R$, and  $\E$ respectively refer to the oblivious, semi-oblivious, restricted, and equivalent\footnote{
	The equivalent chase, like the better-known core chase, halts exactly when the KB has a finite universal model \cite{DBLP:journals/tplp/DelivoriasLMU21}, while being monotonic $(\forall i, \inst_i \subseteq \inst_{i+1})$.
}
variants.

\begin{definition}[Applicability]
	\label{definition:applicability}
	A trigger $t = \pair{\arule, \ahom}$ on an instance $\inst$ is \emph{$\Ob$-applicable} on $\inst$ if  $\trigoutput(t) \not \subseteq \inst$; \emph{$\SO$-applicable} on $\inst$ if $\trigoutput(t') \not \subseteq \inst$ for every trigger $t' = (R, \pi')$ with $\pi(\frontier(R)) = \pi'(\frontier(R))$; \emph{$\R$-applicable} on $\inst$ if there is no retraction from $\inst \cup \trigoutput(t)$ to $\inst$; \emph{$\E$-applicable} on $\inst$ if there is no homomorphism from $\inst \cup \trigoutput(t)$ to $\inst$.
\end{definition}

\begin{example}
	\label{example:preliminaries}
	Consider the KB $\kb = \pair{\rs, \inst}$ with $\rs = \set{R = \apred(x,y) \rightarrow \exists z ~ \apred(y,z), \apred(z,z)}$ and $\inst = \set{ \apred(a,a) }$ for some $a$.
	The trigger $t_1 = (R, \set{ x \mapsto a, y \mapsto a })$, whose output is $\trigoutput(t_{1}) = \set{ \apred(a, z_{t_1}), \apred(z_{t_1},z_{t_1})}$, is $\Ob$ and $\SO$-applicable on $\inst$ but not $\R$ or $\E$-applicable, since there is a retraction from $\inst \cup \trigoutput(t_{1})$ to $\inst$ that maps $z_{t_1}$ to $a$.
\end{example}

\begin{definition}[$\X$-Chase]
	For an $\X \in \set{\Ob, \SO, \R, \E}$, an \emph{$\X$-derivation from a KB $\kb = \pair{\rs, \inst}$} is
	a derivation $\der$ such that every trigger $t_{i} \in \triggers(\der)$ is  $\X$-applicable on $\inst_i$.
	An $\X$-derivation $\der$ is \emph{fair} if for every $\inst_i$ occurring in $\der$ and $\X$-applicable trigger $t$ on $\inst_i$,
	there is some $j > i$ such that $t$ is not $\X$-applicable on $\inst_{j}$.
	An $\X$-derivation is \emph{terminating} if it is fair and finite.
\end{definition}

The result of any fair $\X$-derivation is a universal model of the KB, for $\X \in \set{\Ob, \SO, \R}$, and has a retraction to a universal model for $\X = \E$. Therefore, we obtain:

\begin{proposition}
	\label{proposition:chase-correctness}
	Consider a BCQ $\query$, a KB $\kb$, and some fair $\X$-derivation $\der$ from $\kb$ where $\X \in \set{\Ob, \SO, \R, \E}$.
	Then, $\kb \models \query$ iff $\res(\der) \models \query$.
\end{proposition}

These chase variants induce abstract classes, whose relationship is summarized in \cite{KR2022-11}.
\begin{definition}[Chase-Terminating Sets]
	\label{definition:chase-terminating-sets}
	For an $\X \in \set{\Ob, \SO, \R, \E}$, let $\chaseterm{\X}{\forall}$ (resp. $\chaseterm{\X}{\exists}$) be the set of all rule sets $\rs$ such that every (resp. some) fair $\X$-derivation
	from every KB $\pair{\rs, \inst}$ is finite.
\end{definition}

\begin{proposition}\label{prop:terminating dependence}
	For all $\X \in \set{\Ob, \SO, \E}$, $\chaseterm{\X}{\forall}=\chaseterm{\X}{\exists}$, and
	$\chaseterm{\Ob}{\forall}\subset\chaseterm{\SO}{\forall}\subset\chaseterm{\R}{\forall}\subset\chaseterm{\R}{\exists}\subset\chaseterm{\E}{\forall}$.
\end{proposition}

\begin{example}
	Consider the KB $\kb = \pair{\rs, \inst}$ from \zcref{example:preliminaries}.
	All fair $\Ob$- or $\SO$-derivations from $\kb$ are infinite.
	The only fair $\R$ and $\E$-derivation from $\kb$ is $\der  = (\emptyset, \inst)$.
	Any fair $\R$-derivation from a KB with $\rs$ is finite and hence, $\rs \in \chaseterm{\R}{\forall}$ (and thus $\rs \in \chaseterm{\E}{\forall}$).
\end{example}

\paragraph{\bts}
The class \bts is defined from the notion of treewidth and contains $\chaseterm{\X}{Q}$ for all $\X\in\set{\Ob,\SO,\R,\E}$ and $Q\in\set{\exists,\forall}$ \cite{DBLP:journals/jair/CaliGK13}.
We only introduce a sufficient condition for a rule set not to be \bts.
Given a binary predicate $\apred$, the \emph{$\apred$-graph} of an instance $\inst$ is the directed graph $(\adom(\inst),E)$ where $(a,b)\in E$ if and only if $\apred(a,b) \in \inst$.
\begin{proposition}\label{prop:inf clique implies not bts}
	If there is an infinite clique in the $\apred$-graph of $\res(\der)$ for some binary predicate $\apred$ and derivation $\der$ from $\kb=\pair{\inst,\rs}$, then $\rs$ is not $\bts$.
\end{proposition}

\subsection{Three-counter Machines}

A \emph{three-counter machine} (3CM) $\threecm$ is essentially a two-counter automaton augmented with a strictly increasing time counter. Formally, it is a tuple $\pair{Q, q_1, \delta}$ where $Q$ is a finite set of states including the \emph{initial state} $q_1$, and $\delta:Q\times\set{0,1}^2\to Q\times\set{-1,0,1}^2$ is a transition function.
A \emph{configuration} of $\threecm$ is a tuple $\pair{q, v_1, v_2, t}$ where $q \in Q$ is the current state and $v_1, v_2, t \in \mathbb{N}$ are the values of the three counters.
The \emph{initial configuration} of $\threecm$ is $\pair{q_1, 0, 0, 0}$.
Given a configuration $C=\pair{q, v_1, v_2, t}$, we define $\nextconfig(C)=\pair{q', v_1', v_2', t+1}$ where $\delta(q, b_1, b_2) = \pair{q', d_1, d_2}$ where $b_i = 0$ if $v_i = 0$ and $b_i = 1$ otherwise, and $v_i' = \max(v_i + d_i, 0)$ for $i \in \set{1,2}$.
Note that $\nextconfig$ is a partial function since $\delta$ may be undefined for some inputs.
A 3CM $\threecm$ \emph{halts} if there is a finite sequence of configurations $C_0, \ldots, C_n$ such that $C_0$ is the initial configuration, $C_{i+1} = \nextconfig(C_i)$ for every $0 \leq i < n$, and $\nextconfig(C_n)$ is undefined.
The halting problem for 3CMs is undecidable \cite{Minsky1967}.
\section{The Rule Set}
\label{sec:ruleset}

\begin{figure*}[ht]
	\begin{minipage}[t]{0.49\linewidth}
		\begin{align}
			\pSuccz(x,y) &\to \pSucc(x,y) \Rulename{R_{\pSucc}}\label{rule:successor}\\
			\pR_i(x,y), \pSucc(y,z) &\to \pR_{i+1}(x,z) \Rulename{R_{\pR, i}}\label{rule:count}\\
			\pT_i(x,y,z), \pSucc^{r_i}(y,y'), \pSucc^{q_i}(z,z') &\to \pT_i(x,y',z') \Rulename{R_{\pT, i}}\label{rule:mult}\\
			\pSucc(x,y) &\to \pFlood(y) \Rulename{R_\mathsf{flood}^\mathsf{prop}}\label{rule:flood propagate}
		\end{align}
	\end{minipage}
	\hfill
	\begin{minipage}[t]{0.49\linewidth}
		\begin{align}
			\pEnd(x) &\to \pSucc(x,x) \Rulename{R_{\pSucc}^\pEnd}\label{rule:successor end}\\
			\pSucc^2(x,y) &\to \pG(x,y) \Rulename{R_{\pG}^\mathsf{init}}\label{rule:G init}\\
			\pG(x,y), \pR_i(x,y), \pT_i(x,y,z) &\to \pG(x,z) \Rulename{R_{\pG, i}^\mathsf{step}}\label{rule:G step}
		\end{align}
	\end{minipage}
	\vspace{-0.2cm}
	\begin{equation}
		\pFlood(x), \pFlood(y), \pFlood(z) \to \pG(x,y), \bigwedge_{j=0}^{p-1} \pR_j(x,y), \pT_j(x,y,z) \Rulename{R_\mathsf{flood}^\mathsf{gen}}\label{rule:flood generate}
	\end{equation}
	\vspace{-0.3cm}
	\begin{equation}
		\pG(y,z), \pEnd(z) \to \exists x~\pSuccz(x,y),\pR_0(x,x),\bigwedge_{j=0}^{p-1} \pT_j(x,x,x) \Rulename{R_\exists}\label{rule:existential}
	\end{equation}
	\vspace{-0.4cm}
	\caption{Rules in $\rsReduc$ for all $i<p$. In Rule~\ref{rule:count}, $\pR_{p}$ is interpreted as $\pR_{0}$. $\pSucc^{n}(x,y)$ denotes an $\pSucc$-path of length $n$ from $x$ to $y$.}
	\label{fig:ruleset}
\end{figure*}

We define the class $\Cmc=\set{\rsReduc}[\threecm\text{ a 3CM}]$, using the rule set $\rsReduc$ given in \zcref{fig:ruleset}. As it is recognizable in linear time, it is indeed a concrete class.
The rule set $\rsReduc$ is a adapted from the one used to prove Theorem 1 in \cite{DBLP:conf/icalp/GogaczM14}, stating that all-instance chase termination is undecidable for the oblivious chase.
It provides a reduction from the halting problem for 3CMs to chase termination.

Roughly speaking, the chase on KB $\pair{\set{\pEnd(w)}, \rsReduc}$ looks like a sequence of elements connected through a binary predicate $\pSucc$, ending in the ``well of positivity'' $w$. 
Each element in the chase runs its own private computation of $\threecm$, and is only allowed to generate a predecessor if the simulation of $\threecm$ goes on for long enough. Consequently, this chain of elements is infinite if and only if $\threecm$ does not halt. 
To elaborate, we first need to explain how $\threecm$ can be simulated by a sequence of natural numbers.

Following Section~3.3 in \cite{DBLP:conf/icalp/GogaczM14}, there is a way to encode configurations of a 3CM $\threecm$ as natural numbers using prime factorization.
Let $p_1, p_2, \ldots$ be the sequence of prime numbers, and $Q=\set{q_1,\dots,q_m}$.
Then, a configuration $\pair{q_i, v_1, v_2, t}$ of $\threecm$ is encoded as the natural number
\[\enc(q_i, v_1, v_2, t) = p_i \cdot p_{m+1}^{v_1} \cdot p_{m+2}^{v_2} \cdot p_{m+3}^{t}.\]
In particular, the initial configuration $\pair{q_1, 0, 0, 0}$ is encoded as $2 = p_1$.
Let $p=p_1\dots p_{m+3}$.
\begin{proposition}\label{there are numbers to go from a config to the next}
	There exist natural numbers $q_i, r_i$ for all $1 \leq i \leq p$ such that $\frac{q_i}{r_i}$ is irreducible, and for any configuration $C$ of $\threecm$ such that $i=\enc(C) \bmod p$, if $\nextconfig(C)$ is defined, then ${\enc(\nextconfig(C)) = \frac{q_i}{r_i} \cdot \enc(C)}$,
	and ${\frac{q_i}{r_i}\cdot\enc(C)=\enc(C)}$ otherwise.
\end{proposition}
For further details, we refer the reader to Theorem~5 in \cite{DBLP:conf/icalp/GogaczM14}. Let $\nextenc$ be the function defined as $\nextenc(n) = \frac{q_{n \bmod p}}{r_{n \bmod p}} \cdot n$. Then, $\nextconfig(C)$ is defined if and only if $\nextenc(\enc(C)) \neq \enc(C)$, and in that case $\enc(\nextconfig(C)) = \nextenc(\enc(C))$. Let $\encSet=\set{\nextenc^i(2)}[i \in \nat]$. Since the third counter always increases when a subsequent configuration exists, we get the following.
\begin{proposition}\label{prop:M halts iff G bdd}
	$\threecm$ halts \emph{iff} $\encSet$ is bounded.
\end{proposition}

Then, as mentioned earlier, the chase is shaped as a sequence of elements that form a chain connected by predicates $\pSucc$, where $\pSucc(x,y)$ means that $y$ thinks it is the successor of $x$. The use of ``thinks'' is intentional, as we run a simulation of $\threecm$ from the point of view of each element. 
All predicates used for the simulation ($\pG$, $\pR_i$ and $\pT_i$) feature this element in first position, which thinks itself as $0$. 
The atom $\pG(x,y)$ signifies that from the point of view of $x$, the number represented by $y$ belongs to $\encSet$. This relation is computed according to \zcref{there are numbers to go from a config to the next}, using Rules~\ref{rule:G init} and \ref{rule:G step}. Specifically, the initial configuration is encoded by 2, represented by the second successor of each element (Rule~\ref{rule:G init}). Rule~\ref{rule:G step} then implements the transition to the next configuration using predicates $\pR_i$ and $\pT_i$. The atom $\pR_i(x,y)$ indicates that, from $x$'s perspective, the number represented by $y$ has a remainder of $i$ modulo $p$, as generated by Rule~\ref{rule:count}. Similarly, $\pT_i(x,y,z)$ implies that $\frac{y}{z}=\frac{r_i}{q_i}$ from $x$'s perspective. This is ensured by Rule~\ref{rule:mult}, which implements multiplication by successive additions. These predicates are initialized for each element $x$ by Rule~\ref{rule:existential}. This rule allows an element $y$ to create a predecessor $x$ (specifically $\pSuccz(x,y)$, which implies $\pSucc(x,y)$ via Rule~\ref{rule:successor}) if the simulation of $\threecm$ from $y$ reaches the well of positivity $w$ (i.e., when $\pG(y,z)$ and $\pEnd(z)$ hold for some $z$). This process is illustrated in \zcref{fig:fig4-rule-set-expl} for predicate $\pT_i$.

The reduction described so far mirrors the one in \cite{DBLP:conf/icalp/GogaczM14}. However, the same reduction does not work here. First, it applies only to the oblivious chase; the restricted and equivalent chases always terminate, even if $\threecm$ does not halt. Second, arguing that BCQ entailment is decidable is difficult (and likely false) because the predicates $\pR_i$, $\pT_i$, and $\pG$ lack structure. We introduce two modifications to address these issues.

First, to accommodate all chase variants, we introduce predicates $\pSuccz$ and $\pEnd$, which derive $\pSucc$ via Rules~\ref{rule:successor} and \ref{rule:successor end}. In particular, $\pSuccz$ has no self-loop over the well of positivity $w$ (where $\pEnd(w)$ holds, at the end of the chain), preventing any attempt at homomorphically embedding the whole chain into $w$. Thus, even the equivalent chase is forced to create an infinite chain if $\threecm$ does not halt.

Second, to make BCQ entailment decidable, we employ a flooding mechanism using the predicate $\pFlood$, with Rules~\ref{rule:flood generate} and \ref{rule:flood propagate}. This ensures that once the simulation from $x$ reaches $w$ and generates a predecessor, $x$ connects to all its successors via $\pR_i$, $\pT_i$, and $\pG$. As a result, all elements whose simulation reaches the well of positivity show the same behavior regarding BCQ answering, except for predicates $\pSuccz$, $\pSucc$ and $\pEnd$ which are easier to handle.

{
\pgfmathsetmacro{\nodeSep}{1.5} 

\begin{figure*}
\begin{center}
\begin{tikzpicture}[
	boxedNode/.style={draw, rectangle, minimum size=0.6cm, rounded corners},
   	Spred/.style={-stealth, thick},
   	partTpred/.style={thick, purple},
   	Tpred/.style={partTpred, -Stealth},
   	predLabel/.style={fill=white, rectangle, inner sep=1pt, semithick},
   	xscale=1.5
]
\node [boxedNode] (t0) at (0,0) {$t_0$};
\node [boxedNode] (t1) at (\nodeSep,0) {$t_1$};
\node [boxedNode] (t2) at (2*\nodeSep,0) {$t_2$};
\node [boxedNode] (t3) at (3*\nodeSep,0) {$t_3$};
\node [boxedNode] (t4) at (4*\nodeSep,0) {$t_4$};
\node [boxedNode] (w) at (5*\nodeSep,0) {$w$};

\node [] (wEnd) at (5.2*\nodeSep,-0.3*\nodeSep) {$\pEnd$};

\path [Spred] (t0) edge node[predLabel, pos=0.4] {\small$\pSucc$} (t1);
\path [Spred] (t1) edge node[predLabel, pos=0.4] {\small$\pSucc$} (t2);
\path [Spred] (t2) edge node[predLabel, pos=0.4] {\small$\pSucc$} (t3);
\path [Spred] (t3) edge node[predLabel, pos=0.4] {\small$\pSucc$} (t4);
\path [Spred] (t4) edge node[predLabel, pos=0.4] {\small$\pSucc$} (w);
\path [Spred, out=30, in=330, looseness=5] (w) edge node[predLabel] {$\pSucc$} (w);

\path [partTpred, out=140, in=180, looseness=10] (t0) edge node[above] {\small Step 0} (t0);
\path [Tpred, out=180, in=220, looseness=10] (t0) edge node[predLabel, pos=0.55] {$\pT_i$} (t0);
\path [partTpred, out=45, in=90, looseness=0.7] (t0) edge node[above, pos=.65] {\small Step 1} (t2);
\path [Tpred, out=90, in=90] (t2) edge node[predLabel] {$\pT_i$} (t3);
\path [partTpred, out=315, in=270, looseness=0.5] (t0) edge node[pos=.15, below] {\small Step 2} (t4);
\path [Tpred, out=270, in=270] (t4) edge node[predLabel] {$\pT_i$} (w);
\path [partTpred, out=65, in=115, looseness=0.6] (t0) edge node[pos=.15, above] {\small Step 3} (w);
\path [Tpred, out=115, in=65, looseness=10] (w) edge node[predLabel] {$\pT_i$} (w);

\end{tikzpicture}
\end{center}

 \caption{Representation of predicates $\pSuccz$, $\pEnd$ and atoms with shape $\pT_i(t_0,\_,\_)$ in $\stepXend{5}$ for some $i$ such that $\frac{r_i}{q_i}=\frac{2}{3}$. The $\pSuccz$-atoms, as described in \zcref{lem:chase is a chain}, form a chain ending in $w$, the only term with $\pEnd(w)$. Then, $\pT_i$-atoms are as described by \zcref{lem:the predicates compute at least their semantics}: First, $\pT_i(t_0,t_0,t_0)$ is created by the application of Rule~\eqref{rule:existential} that creates $t_0$ (step 0). Then, $\pT_i(t_0,t_2,t_3)$, $\pT_i(t_0,t_4,t_6)$, and $\pT_i(t_0,t_6,t_9)$ (note that $w=t_6=t_9$) are created by successive applications of Rule~\eqref{rule:mult} (steps 1, 2 and 3, respectively).}
 \label{fig:fig4-rule-set-expl}
\end{figure*}
}


\section{Undecidability of Chase Termination and \bts-Recognition}
\label{sec:chase-termination}

This section is devoted to Points~1 and 2 of \zcref{theorem:ultimate-goal}, which are the negative properties of our class: chase termination for all chase variants and \bts are both undecidable in $\Cmc$.
Regarding Oblivious and Semi-Oblivious chase termination, this is mostly guaranteed \emph{by construction}, as we adapted the class proposed in \cite{DBLP:conf/icalp/GogaczM14} that already had these properties. Remarkably, our set of rules also exhibits this behavior for the equivalent chase. To prove this, we make use of \zcref{prop:terminating dependence} as much as possible. More precisely, we show that if $\threecm$ does not halt, then the equivalent chase does not terminate on $\rsReduc$ and some specific instance $\IEnd$ (\zcref{thm:if G ubdd then chase does not terminate}), and that if $\threecm$ does halt, then the oblivious chase terminates on the so-called critical instance, and thus on all instances (\zcref{thm:if G bdd then chase terminates}). These two results, along with \zcref{prop:terminating dependence} and the undecidability of the halting problem for 3CMs entail that chase termination is undecidable for $\Cmc$.
The undecidable \bts recognition is then a consequence of the above, joint with our freshly added flooding mechanism.

In this section, $\X$ always denotes an element of $\set{\Ob,\SO,\R,\E}$. We denote the result of some fair $\X$-derivation from $\inst,\rs$ with $\chase^\X(\inst,\rs)$. Note that this object is not unique in general, as different derivations may produce different results. This is however not an issue for our purpose, since we do not really consider the restricted chase in proofs, and it is the only variant for which the order of application of triggers matters.
It will also be convenient to consider the chase as an alternation of closing under Datalog rules and then under non-Datalog rules.
We define $\stepX{0}(\inst)=\chase^\X(\inst,\DatS{\rsReduc})$ and for all $i \geq 0$, $\stepX{i+1}(\inst)=\chase^\X(\chase^\X(\stepX{i}(\inst),\NDatS{\rsReduc}), \DatS{\rsReduc})$, where $\DatS{\rs}$ and $\NDatS{\rs}$ denote the sets of Datalog and non-Datalog rules in $\rs$, respectively.
Note that this process still induces a fair derivation.

\subsection{Non-Termination of the Chase}
\label{sec:non-termination}

We first prove that if $\threecm$ does not halt, then $\rsReduc\notin\chaseterm{\X}{Q}$ for $Q\in\set{\forall, \exists}$, by considering only the instance $\IEnd=\set{\pEnd(w)}$. Before proving this, we need two lemmas about the structure of the chase.

\begin{lemmarep}\label{lem:chase is a chain}
	For all $n>0$, if $\stepXend{n}\neq\stepXend{n-1}$, then $\adom(\stepXend{n})=\set{t_0, \dots, t_{n-1}, w}$ for some terms $t_0, \dots, t_{n-1}$, and the $\pSuccz$-atoms of $\stepXend{n}$ are exactly $\set{\pSuccz(t_i,t_{i+1})}[0 \leq i < n]$ where $t_n=w$.
\end{lemmarep}
\begin{proofsketch}
	By induction on $n$. Both the base case and inductive step follow from the fact that there is a single existential rule, and thus the only way for $\stepXend{n}$ to be different from $\stepXend{n-1}$ is for Rule~\ref{rule:existential} to be applied. This rule features a single frontier variable, which has to be mapped to the term created at the previous step.
\end{proofsketch}
\begin{proof}
	By induction on $n$. For the base case ($n=1$), notice that $\stepXend{0}=\set{\pEnd(w),\pSucc(w,w),\pFlood(w),\pG(w,w)}\cup\set{\pR_i(w,w),\pT_i(w,w,w)}[i<p]$ by Rules~\ref{rule:successor end}, \ref{rule:flood propagate} and \ref{rule:flood generate}. Thus, if $\stepXend{1}\neq\stepXend{0}$, Rule~\ref{rule:existential} (the only existential rule) must have been applied, and created $\pSuccz(t_0,w)$ for some fresh term $t_0$, so $\adom(\stepXend{1})=\set{t_0,w}$. Since no Datalog rule can yield $\pSuccz$-atoms, the only $\pSuccz$-atom is $\pSuccz(t_0,w)$, as required.
	
	For the inductive step, assume that the result holds for some $n\geq 1$ and that $\stepXend{n+1}\neq\stepXend{n}$. By the induction hypothesis, there are terms $t'_0, \dots, t'_{n-1}$ such that $\adom(\stepXend{n})=\set{t'_0, \dots, t'_{n-1}, w}$ and the $\pSuccz$-atoms of $\stepXend{n}$ are exactly $\set{\pSuccz(t'_i,t'_{i+1})}[0 \leq i < n]$.

	As there is a single existential rule in $\rsReduc$, all terms $t'_i$ have been generated by Rule~\ref{rule:existential} where $y$ is mapped to $t'_{i+1}$ (and to $w$ for $t'_{n-1}$). Thus, the only way for $\stepXend{n+1}$ to differ from $\stepXend{n}$ is for Rule~\ref{rule:existential} to be applied by mapping $y$ to $t'_0$, creating a new term $t_0$ and the atom $\pSuccz(t_0,t'_0)$. Let $t_i=t'_{i-1}$ for all $1\leq i \leq n$, so that $\adom(\stepXend{n})=\set{t_1, \dots, t_n, w}$.
	
	Since no Datalog rule can yield $\pSuccz$-atoms, the only $\pSuccz$-atoms of $\stepXend{n+1}$ are exactly $\set{\pSuccz(t_i,t_{i+1})}[0 \leq i < n+1]$, as required.
\end{proof}

In the following lemma, proven by a simple induction, we refer to terms of $\stepXend{n}$ as $t_0, \dots, t_{n-1}$ as in \zcref{lem:chase is a chain}. We also use $t_i$ for $i\geq n$ to refer to $w$. Note that even for such $i$'s, $\pSucc(t_i,t_{i+1})\in\stepXend{n}$.

\begin{lemmarep}\label{lem:the predicates compute at least their semantics}
	For all $n>0$, if $\stepXend{n}\neq\stepXend{n-1}$,
	\begin{enumerate}[label=(\roman*)]
		\item $\pR_i(t_0,t_k)\in\stepXend{n}$ if and only if $k=i\mod p$; \label{lemitem:semantics:R}
		\item $\pT_i(t_0,t_k,t_l)\in\stepXend{n}$ if and only if there is $m$ such that $k=mr_i$ and $l=mq_i$; \label{lemitem:semantics:T}
		\item For all $k\in\nat$, $\pG(t_0,t_{g^k(2)})\in\stepXend{n}$; \label{lemitem:semantics:G}
		\item If $\encSet$ is bounded by $n-1$, then for all $l$ such that $\pG(t_0,t_l)\in\stepXend{n}$, there is $k$ such that $l = g^k(2)$. \label{lemitem:semantics:G only}
	\end{enumerate}
\end{lemmarep}
\begin{proof}
	For \zcref{lemitem:semantics:R}, we proceed by induction on $k$. Initially, when $t_0$ is created by Rule~\ref{rule:existential}, there is only $\pR_0(t_0,t_0)$. Then, $\pFlood(t_0)\notin\stepXend{n}$ by \zcref{lem:chase is a chain} and the fact that only Rule~\ref{rule:flood propagate} can create $\pFlood$-atoms. Thus, only Rule~\ref{rule:count} can generate new $\pR_i$-atoms with $t_0$ as first argument. By induction, assume that $k-1=i\mod p$, and thus that $\pR_i(t_0,t_{k-1})\in\stepXend{n}$, for some $k\in\nat$. Then, since $\pSucc(t_{k-1},t_k)\in\stepXend{n}$, Rule~\ref{rule:count} creates $\pR_{(i+1)\mod p}(t_0,t_{k+1})$, as required. Since no other rule can create $\pR_i$-atoms with $t_0$ as first argument, \zcref{lemitem:semantics:R} holds.

	For \zcref{lemitem:semantics:T}, we proceed by induction on $m$. Initially, when $t_0$ is created by Rule~\ref{rule:existential}, there is only $\pT_i(t_0,t_0,t_0)$ for all $i<p$. Then, as before, $\pFlood(t_0)\notin\stepXend{n}$, so only Rule~\ref{rule:mult} can generate new $\pT_i$-atoms with $t_0$ as first argument. By induction, assume that $k=mr_i$ and $l=mq_i$, for some $i<p$ and $k,l$. Then, since $\pSucc^{r_i}(t_k,t_{k+r_i})$ and $\pSucc^{q_i}(t_l,t_{l+q_i})$ belong to $\stepXend{n}$, Rule~\ref{rule:mult} creates $\pT_i(t_0,t_{k+r_i},t_{l+q_i})$. As $k+r_i=(m+1)r_i$ and $l+q_i=(m+1)q_i$, and no other rule can create $\pT_i$-atoms with $t_0$ as first argument, \zcref{lemitem:semantics:T} holds.

	For \zcref{lemitem:semantics:G}, we proceed by induction on $k$. For the base case ($k=0$), Rule~\ref{rule:G init} creates $\pG(t_0,t_{g^0(2)})=\pG(t_0,t_2)$. For the inductive step, assume that the result holds for some $k\in\nat$ and let $l=g^k(2)$. By the induction hypothesis, $\pG(t_0,t_l)\in\stepXend{n}$. Then, let $l'=g^{k+1}(2)=g(l)=l\frac{q_i}{r_i}$ for $i=l\mod p$. By Item~(i) above, $\pR_i(t_0,t_l)\in\stepXend{n}$. Then, since $\frac{l}{l'}=\frac{r_i}{q_i}$, there is some $m$ such that $l=mr_i$ and $l'=mq_i$. Thus, by Item~(ii) above, $\pT_i(t_0,t_l,t_{l'})\in\stepXend{n}$. Thus, by Rule~\ref{rule:G step}, $\pG(t_0,t_{l'})=\pG(t_0,t_{g^{k+1}(2)})\in\stepXend{n}$, as required.

	For \zcref{lemitem:semantics:G only}, we prove that if $\encSet$ is bounded by $n-1$, then the only $\pG$-atom that can be generated by Rule~\ref{rule:G step} by mapping $y$ to $t_{g^k(2)}$ is $\pG(t_0,t_{g^{k+1}(2)})$. Since $\encSet$ is bounded by $n-1$, by Item~(i), the only $i$ such that $\pR_i(t_0,t_{g^k(2)})\in\stepXend{n}$ is $i=g^k(2)\mod p$. Then, by Item~(ii), the only $\pT_i$-atom that features $t_{g^k(2)}$ in second position is $\pT_i(t_0,t_{g^k(2)},t_{g^{k+1}(2)})$. Thus, the only $\pG$-atom that can be generated by Rule~\ref{rule:G step} by mapping $y$ to $t_{g^k(2)}$ is $\pG(t_0,t_{g^{k+1}(2)})$, as required. Thus, by induction on $k$, the only $\pG$-atoms with $t_0$ as first argument are those of the form $\pG(t_0,t_{g^k(2)})$, as required.
\end{proof}

\begin{theoremrep}\label{thm:if G ubdd then chase does not terminate}
	If $\threecm$ does not halt, there is no terminating $\X$-chase sequence from $\pair{\set{\pEnd(w)},\rsReduc}$.
\end{theoremrep}%
\begin{proofsketch}
	We simply prove by induction on $n$ that if $\threecm$ does not halt and $\stepXend{n}\neq\stepXend{n-1}$, then Rule~\ref{rule:existential} is $\E$-applicable on $\stepXend{n}$. In this case, $\stepXend{n}$ is as described in \zcref{lem:chase is a chain}. Then, by \zcref{prop:M halts iff G bdd}, $\encSet$ is unbounded, so $\pG(t_0,w)\in\stepXend{n}$ by \zcref{lemitem:semantics:G only} of \zcref{lem:the predicates compute at least their semantics}. As the $\pSuccz$-atoms of $\stepXend{n}$ are exactly $\set{\pSuccz(t_i,t_{i+1})}[0 \leq i < n]$, there is no way to homomorphically map the result of applying \ref{rule:existential} into $\stepXend{n}$, so this rule is indeed $\E$-applicable.
\end{proofsketch}
\begin{proof}[Proof of \zcref{thm:if G ubdd then chase does not terminate}]
	We prove by induction on $n$ that if $\encSet$ is unbounded (which coincides with $\threecm$ not halting by \zcref{prop:M halts iff G bdd}), then Rule~\ref{rule:existential} is $\E$-applicable (and thus $\X$-applicable for all variants $\X$) on $\stepXend{n}$.

	For the base case, note that $\stepXend{0}=\set{\pEnd(w),\pSucc(w,w),\pFlood(w),\pG(w,w)}\cup\set{\pR_i(w,w),\pT_i(w,w,w)}[i<p]$, as highlighted in the proof of \zcref{lem:chase is a chain}. Thus, Rule~\ref{rule:existential} is applicable on $\stepXend{0}$. Since there is no $\pSuccz$-atom in $\stepXend{0}$, there is no way to homomorphically map the fresh term created by Rule~\ref{rule:existential} to any existing term, so the rule is indeed $\E$-applicable.

	For the inductive step, assume that the result holds for some $n\geq 0$. Thus, $\stepXend{n+1}\neq\stepXend{n}$ by the induction hypothesis, so by \zcref{lem:chase is a chain}, there are terms $t_0, \dots, t_n$ such that $\adom(\stepXend{n+1})=\set{t_0, \dots, t_n, w}$ and the $\pSuccz$-atoms of $\stepXend{n+1}$ are exactly $\set{\pSucc(t_i,t_{i+1})}[0 \leq i < n+1]$. Thus, by Rule~\ref{rule:successor}, for all $k>n+1$, $\pSucc^k(t_0,w)\in\stepXend{n+1}$. Since $\encSet$ is unbounded, there is some $k$ such that $g^k(2)>n+1$. Hence, by \zcref{lemitem:semantics:G} of \zcref{lem:the predicates compute at least their semantics}, $\pG(t_0,w)\in\stepXend{n+1}$. Then, since $\pEnd(w)\in\stepXend{n+1}$, Rule~\ref{rule:existential} is applicable on $\stepXend{n+1}$ by mapping $y$ to $t_{g^k(2)}$.
	Since by construction, the $\pSuccz$-atoms of $\stepXend{n+1}$ are exactly $\set{\pSuccz(t_i,t_{i+1})}[0 \leq i < n]$, there is no way to homomorphically map the result of applying \ref{rule:existential} into $\stepXend{n+1}$, so the rule is indeed $\E$-applicable.

	Thus, this proves that at each step, Rule~\ref{rule:existential} is $\E$-applicable, so the $\E$-chase does not terminate on $\pair{\set{\pEnd(w)},\rsReduc}$. This entails that no $\X$-chase sequence from $\pair{\set{\pEnd(w)},\rsReduc}$ terminates, as the $\E$-chase terminates if and only if a universal model exists \cite{chase_revisited}, and any other variant terminating would yield such a model.
\end{proof}

\subsection{Termination of the Chase}
\label{sec:termination}

To prove that the $\X$-chase terminates if $\threecm$ halts (\zcref{thm:if G bdd then chase terminates}), it suffices to prove that in this case, the oblivious chase terminates on the so-called critical instance. The critical instance is the instance containing a single constant $w$, and in which all possible facts over the schema of the rule set are present. Indeed, if the oblivious chase terminates on the critical instance, it terminates on all instances \cite{generalized_schema_mappings}, and thus all chase variants terminate too via \zcref{prop:terminating dependence}. Let $\critstep{}$ be the critical instance for the schema of $\rsReduc$, and $\critstep{n}=\stepX{n}(\critstep{})$.

\begin{lemmarep}\label{lem:critical instance is like end instance}
	For all $n\in\nat$, $\critstep{n}\simeq\stepOend{n}\cup\set{\pSuccz(w,w)}$.
\end{lemmarep}%
\begin{proofsketch}
	By induction on $n$. The base case is obtained using Rules~\ref{rule:successor end}, \ref{rule:flood propagate} and \ref{rule:flood generate}. For the inductive step, the additional $\pSuccz$-atom of $\critstep{n}$ can only produce $\pSucc(w,w)$, which already belongs to $\critstep{}$. Thus, both instances feature the same applicable triggers.
\end{proofsketch}
\begin{proof}
	By induction on $n$. For the base case ($n=0$), note that $\stepOend{0}=\set{\pEnd(w),\pSucc(w,w),\pFlood(w),\pG(w,w)}\cup\set{\pR_i(w,w),\pT_i(w,w,w)}[i<p]$ by Rules~\ref{rule:successor end}, \ref{rule:flood propagate} and \ref{rule:flood generate}. Since the critical instance contains all possible facts over the schema, $\critstep{0}=\critstep{}=\stepOend{0}\cup\set{\pSuccz(w,w)}$, as required.

	For the inductive step, assume that the result holds for some $n\geq 0$. Thus, $\critstep{n}\simeq\stepOend{n}\cup\set{\pSuccz(w,w)}$. The only possible trigger that can use $\pSuccz(w,w)$ uses Rule~\ref{rule:successor} to produce $\pSucc(w,w)$, which is already present in $\critstep{}$. Thus, all other triggers are either applicable on both instances or on none of them, and produce the same atoms, so $\critstep{n+1}\simeq\stepOend{n+1}\cup\set{\pSuccz(w,w)}$, as required.
\end{proof}
\begin{theoremrep}\label{thm:if G bdd then chase terminates}
	If $\threecm$ halts, then the $\X$-chase terminates on $\pair{\inst,\rsReduc}$ for all instances $\inst$.
\end{theoremrep}%
\begin{proofsketch}
	If $\threecm$ halts, then $\encSet$ is bounded by $n$ and by \zcref{lemitem:semantics:G,lemitem:semantics:G only} of \zcref{lem:the predicates compute at least their semantics}, $\pG(t_0,w)\notin\critstep{n+1}$. Thus, Rule~\ref{rule:existential} is not applicable on $\critstep{n+1}$, which entails the result as discussed before \zcref{lem:critical instance is like end instance}.
\end{proofsketch}
\begin{proof}[Proof of \zcref{thm:if G bdd then chase terminates}]
	We use the notations of \zcref{lem:chase is a chain} to designate the terms of $\critstep{n}$, since according to \zcref{lem:critical instance is like end instance}, it is isomorphic to $\stepOend{n}\cup\set{\pSuccz(w,w)}$. Assume that $\encSet$ is bounded by $n\in\nat$. By \zcref{lemitem:semantics:G,lemitem:semantics:G only} of \zcref{lem:the predicates compute at least their semantics}, all $\pG$-atoms of $\critstep{n+1}$ with $t_0$ as first argument are of the form $\pG(t_0,t_{g^k(2)})$ for some $k$, such that $g^k(2)\leq n$. Thus, in particular, $\pG(t_0,w)\notin\critstep{n+1}$, so Rule~\ref{rule:existential} is not applicable on $\critstep{n+1}$. Since it is the only existential rule in $\rsReduc$, the $\Ob$-chase terminates on $\pair{\critstep{0},\rsReduc}$ by step $n+1$. Thus, by the initial remarks of this section, the $\X$-chase terminates on $\pair{\inst,\rsReduc}$ for all instances $\inst$.
\end{proof}
\subsection{Treewidth}
\label{sec:bounded treewidth}

We then prove that $\bts$-membership in $\Cmc$ is undecidable.

\begin{theorem}\label{thm:bts iff G bdd}
	$\rsReduc$ is $\bts$ if and only if $\threecm$ halts.
\end{theorem}

If $\threecm$ halts, then by \zcref{thm:if G bdd then chase terminates}, the $\X$-chase terminates, so the rule set is clearly $\bts$. The converse direction is a direct consequence of the flooding mechanism.

\begin{lemma}\label{lem:G ubdd implies big clique}
	If $\threecm$ does not halt, then there is an infinite clique in the $\pR_0$-graph of $\chase^\X(\set{\pEnd(w)},\rsReduc)$.
\end{lemma}
\begin{proof}
	Assume $\threecm$ does not halt.
	By \zcref{thm:if G ubdd then chase does not terminate}, for all $n$, $\stepXend{n}\neq\stepXend{n-1}$, so by \zcref{lem:chase is a chain}, there is an infinite sequence of terms $u_0, u_1, \dots$ such that $u_0=w$, and for all $i \geq 0$, $\pSucc(u_{i+1},u_i)\in\chase^\X(\set{\pEnd(w)},\rsReduc)$.
	By Rule~\ref{rule:flood propagate}, for all $i,j \geq 0$, $\pFlood(u_i)\in\chase^\X(\set{\pEnd(w)},\rsReduc)$, so by Rule~\ref{rule:flood generate}, $\pR_0(u_i,u_j)$ is also in the chase result, so the $\pR_0$-graph contains an infinite clique.
\end{proof}

\zcref{lem:G ubdd implies big clique,prop:inf clique implies not bts} conclude the proof.
\section{Decidability of BCQ Entailment}
\label{sec:bcq-entailment}

We now turn to decidability of BCQ entailment, that is Point~3 of Theorem~\ref{theorem:ultimate-goal}.
Consider an instance $\inst$, a set of rules $\rsReduc$ from our class $\Cmc$ constructed in Section~\ref{sec:ruleset}, and a BCQ $\query$.
To decide whether $\pair{\inst, \rsReduc} \models \query$, our algorithm computes the result of the semi-oblivious chase up to depth $\stopthecount \coloneqq 2\size{\query} + 1$, that is $\stepSO{\stopthecount}(\inst)$, and then checks whether $\stepSO{\stopthecount}(\inst) \models \query$.
This clearly terminates as we only require to apply finitely many rules in the first step, and then evaluate a usual BCQ on a finite instance in the second step.
It remains to prove the correctness:
\begin{lemma}\label{lem:no need to go far}
	$\pair{\inst, \rsReduc} \models \query$ iff $\stepSO{\stopthecount}(\inst) \models \query$.
\end{lemma}
The rest of this section is devoted to proving the above.
We first define $\inst_n \coloneqq \stepSO{n}(\inst)$ for all $n$, and $\inst_\infty \coloneqq \bigcup_{n \geq 0} \inst_n$.
The successive $\inst_n$ again induce a fair derivation, therefore $\inst_\infty$ is a universal model of $\pair{\inst, \rsReduc}$. Thus, \zcref{lem:no need to go far} can be rephrased as $\inst_\infty\models \query$ iff $\inst_\stopthecount\models\query$ using \zcref{proposition:chase-correctness}. With this formulation, the `if' direction is the easiest one: indeed, if $\inst_M \models \query$, we get $\inst_\infty\models\query$ from the fact that $\inst_\infty$ contains $\inst_M$.

We now turn to proving the `only-if' direction of Lemma~\ref{lem:no need to go far}.
Assume $\inst_\infty \models \query$. Compactness guarantees that there exists $N$ such that $\inst_N \models \query$.
If $N \leq \stopthecount$, we are done since $\inst_N \subseteq \inst_M$.
Otherwise, if $N > \stopthecount$, consider a homomorphism $\ahom$ from $\query$ to $\inst_N$.
We want to construct a homomorphism $\ahom' : \query \rightarrow \inst_\stopthecount$ based on $\ahom$.
Our main issue is atoms that $\ahom$ maps to $\inst_N\setminus\inst_\stopthecount$, for which we need to find replacements within $\inst_M$.

Intuitively, we distinguish between the $\pEnd$-atoms (which cannot be introduced during the chase), the $\pSuccz$-atoms, the $\pSucc$-atoms, and the others.
Indeed, apart from within the original instance $\inst$, the $\pSuccz$- and $\pSucc$-atoms form chains ending on some element from $\adom(\inst)$.
Therefore, if $h$ involves such $\pSuccz$- and $\pSucc$-atoms from $\inst_N \setminus \inst_M$, we know they lie along a chain, say at distance $\ell$ to $\inst$.
We can then prove that a sufficient part of this chain already belongs to $\inst_M$, due to $M \geq 2\size{q}$.
To build $h'$, we can thus `shift' $h$ by $\ell- M$ atoms to use instead this part of the chain.
Furthermore, we can establish that \emph{every} term from $\inst_N$ lies on such a chain, therefore, when we shift our homomorphism along the chain, due to the combination of Rules~\ref{rule:flood propagate} and \ref{rule:flood generate}, all the $\apred$-atoms are present for $\apred \notin \{ \pEnd, \pSuccz, \pSucc\}$.
The only exception is when such a $\apred$-atom features a term $t$ that is the starting point of a chain (thus Rule~\ref{rule:flood propagate} may not trigger on $t$) and that $\ell \leq M$ (so that we don't shift).
But $\ell \leq M$ guarantees the atom already belongs to $\inst_M$.

We summarize considerations on $\pSuccz$- and $\pSucc$-atoms in the next lemma, whose proof is a simple induction:
\begin{lemmarep}\label{lem:chase is a bunch of chains}
	For all $n \geq 0$ and $a\in\adom(\inst)$, there exist integers $0 \leq m_a \leq n$ s.t.:
	(i) the domain of $\inst_n$ consists of distinct terms $a_i$ with $a \in \adom(\inst)$ and $0 \leq i \leq m_a$, and where $a_0 = a$;
	and (ii)
	$\pSuccz$ atoms of $\inst_n$ are exactly:
	\begin{align*}
		&
		\{ \pSuccz(a, b) \mid \pSuccz(a, b) \in \inst \}
		\\
		&
		\cup \{ \pSuccz(a_{i+1}, a_{i}) \mid
		0 \leq i < m_a, 
		a \in \adom(\inst) \},
	\end{align*}
	and, as a corollary, the $\pSucc$ atoms of $\inst_n$ are exactly: 
	\begin{align*}
		&
		\{ \pSucc(a, b) \mid \pSucc(a, b) \in \inst \text{ or } \pSuccz(a, b) \in \inst \}
		\\
		&
		\cup \{ \pSucc(a,a) \mid \pEnd(a)\in\inst \}
		\\
		&
		\cup \{ \pSucc(a_{i+1}, a_{i}) \mid 0 \leq i < m_a, a \in \adom(\inst) \}.
	\end{align*}
\end{lemmarep}

\begin{proof}
	We proceed by induction on $n \geq 0$.
	
	For $n = 0$, it suffices to set $m_a = 0$ for every $a \in \adom(\inst)$.
	Since we only applied our Datalog rules on $\inst$ (recall that $\inst_0 = \stepX{0}(\inst) = \chase^\SO(\inst,\DatS{\rsReduc})$ by definition), it is clear that Point~($i$) is satisfied: $\adom(\inst_0) = \adom(\inst)$.
	For Point~($ii$), notice that no rule from $\DatS{\rsReduc}$ allows introducing an $\pSuccz$ atom, hence the $\pSuccz$-atoms of $\inst_0$ are exactly $\{ \pSuccz(a, b) \mid \pSuccz(a, b) \in \inst  \}$, as claimed.
	Regarding the $\pSucc$ atoms, notice that the claimed expression exactly accounts for applications of Rules~\ref{rule:successor} and \ref{rule:successor end}.
	
	Let us now assume that the claim holds at some step $n \geq 0$ and let us consider $\inst_{n+1}$.
	We apply the induction hypothesis on $\inst_n$ and denote $m_{n, a}$ the obtained integers, for $a \in \adom(\inst)$.
	As there is a single existential rule in $\rsReduc$, every term $a_i$ of $\adom(\inst_n)$, for $a \in \adom(\inst)$ and $0 < i \leq m_{n, a}$ has been generated by rule \ref{rule:existential} with $y$ necessarily mapped to $a_{i-1}$ due to the description of $\pSuccz$ we have in $\inst_n$.
	Therefore, a trigger for \ref{rule:existential} mapping $y$ on such a term $a_i$ is no longer $\SO$-applicable unless $i = m_{n, a}$.
	In that case, each such application of \ref{rule:existential} produces a new term $a_{m_{n, a}+1}$, along with the atom $\pSuccz(a_{m_{n, a}+1}, a_{m_{n, a}})$.
	Note that there cannot be any further trigger for \ref{rule:existential} on the freshly introduced term $a_{m_{n, a}+1}$ as \ref{rule:existential} does not introduce the necessary $\pG$ atom.
	Therefore, we can set $m_a := m_{n, a} + 1$ whenever such a fresh element is introduced.
	Notice that the induction hypothesis guarantees $m_{n, a} \leq n$, so that $m_a \leq n + 1$ as desired for Point~($i$).
	For Point~($ii$), note that such applications of \ref{rule:existential} are the only way to produce more $\pSuccz$ atoms, so that the description of $\pSuccz$ atoms (and subsequently of $\pSucc$ atoms), is correct.
\end{proof}

\begin{toappendix}
	We now state a direct corollary of Lemma~\ref{lem:chase is a bunch of chains} which provides a similar full-characterization for the predicate $\pFlood$.
	This lemma allows to simplify the presentation of the proof of Lemma~\ref{lem:magic} further in this appendix.
	\begin{corollary}\label{coro:flood}
		Let $n \geq 0$ and consider the terms of $\adom(\inst_n)$ as described in Lemma~\ref{lem:chase is a bunch of chains}.
		The $\pFlood$ atoms of $\inst_n$ are exactly: 
		\begin{align*}
			&
			\{ \pFlood(a) \mid \pFlood(a) \in \inst \text{ or } \pEnd(a)\in\inst \}
			\\
			&
			\cup \{ \pFlood(a) \mid \pSucc(b, a) \in \inst \text{ or } \pSuccz(b, a) \in \inst \}
			\\
			&
			\cup \{ \pFlood(a_{i}) \mid 0 \leq i < m_a, a \in \adom(\inst) \}.
		\end{align*}
	\end{corollary}
	\begin{proof}
		This is a direct consequence of the description of the $\pSucc$-atoms given by Lemma~\ref{lem:chase is a bunch of chains} joint with the observation that a $\pFlood$-atom can only be produced by an application of Rule~\ref{rule:flood propagate}.
	\end{proof}
\end{toappendix}

{
\pgfmathsetmacro{\chainSep}{1.5} 
\pgfmathsetmacro{\betChains}{-1.2} 
\pgfmathsetmacro{\dX}{-0.8*\chainSep}
\pgfmathsetmacro{\dY}{1.25*\betChains}
\pgfmathsetmacro{\eX}{\dX-0.8*\chainSep}
\pgfmathsetmacro{\eY}{.5*\betChains}
\pgfmathsetmacro{\fX}{\eX}
\pgfmathsetmacro{\fY}{1.75*\betChains}

\begin{figure*}
\begin{center}
\begin{tikzpicture}[
	boxedNode/.style={draw, rectangle, minimum size=0.6cm, fill=white, rounded corners},
   	Spred/.style={-stealth, thick},
   	partTpred/.style={thick, purple},
   	Tpred/.style={partTpred, -Stealth},
   	predLabel/.style={fill=white, rectangle, inner sep=1pt, semithick},
   	predLabelColored/.style={fill=magenta!10, rectangle, inner sep=1pt, semithick},
   	xscale=1.5
]

\draw [dashed, thick, magenta, rounded corners, fill=magenta!10] (\eX-0.35*\chainSep, -0.5*\betChains) rectangle (0.35*\chainSep,2.5*\betChains);

\node [anchor=north west] at (\eX-0.25*\chainSep, -0.35*\betChains) {\LARGE $\inst$};

\node [boxedNode] (a) at (0,0) {$a$};
\node [boxedNode] (a1) at (\chainSep, 0) {$a_1$};
\node [boxedNode] (b) at (0,\betChains) {$b$};
\node [boxedNode] (b1) at (\chainSep,\betChains) {$b_1$};
\node [boxedNode] (b2) at (2*\chainSep,\betChains) {$b_2$};
\node [boxedNode] (c) at (0,2*\betChains) {$c$};
\node [boxedNode] (c1) at (\chainSep,2*\betChains) {$c_1$};
\node [boxedNode] (c2) at (2*\chainSep,2*\betChains) {$c_2$};
\node [boxedNode] (c3) at (3*\chainSep,2*\betChains) {$c_3$};

\node (cEnd) at (-0.3*\chainSep,2.25*\betChains) {$\pEnd$};

\node [boxedNode] (d) at (\dX,\dY) {$d$};
\node [boxedNode] (e) at (\eX,\eY) {$e$};
\node [boxedNode] (f) at (\fX,\fY) {$f$};

\path [Spred] (a1) edge node[predLabel, pos=0.4] {\small$\pSuccz$} (a);
\path [Spred] (b1) edge node[predLabel, pos=0.4] {\small$\pSuccz$} (b);
\path [Spred] (b2) edge node[predLabel, pos=0.4] {\small$\pSuccz$} (b1);
\path [Spred] (c1) edge node[predLabel, pos=0.4] {\small$\pSuccz$} (c);
\path [Spred] (c2) edge node[predLabel, pos=0.4] {\small$\pSuccz$} (c1);
\path [Spred] (c3) edge node[predLabel, pos=0.4] {\small$\pSuccz$} (c2);

\path [Spred] (a) edge node[predLabelColored, pos=0.4] {\small$\pSuccz$} (b);
\path [Spred] (b) edge node[predLabelColored, pos=0.4] {\small$\pSuccz$} (c);
\path [Spred, out=250, in=110] (e) edge node[predLabelColored, pos=0.4] {\small$\pSuccz$} (f);
\path [Spred, out=70, in=290] (f) edge node[predLabelColored, pos=0.4] {\small$\pSuccz$} (e);
\path [Spred] (f) edge node[predLabelColored, pos=0.4] {\small$\pSuccz$} (d);
\path [Spred, out=225, in=90] (a) edge node[predLabelColored, pos=0.4] {\small$\pSuccz$} (d);
\path [Spred] (a) edge node[predLabelColored, pos=0.4] {\small$\pSuccz$} (e);
\path [Spred] (c) edge node[predLabelColored, pos=0.4] {\small$\pSuccz$} (f);

\end{tikzpicture}
\end{center}

 \caption{Representation of predicates $\pSuccz$ and $\pEnd$ in $\inst_\infty$ under the hypothesis that $\encSet$ is bounded by $2$. The magenta boxed part of the figure is the instance $\inst$. The $\pSuccz$-atoms, as described in \zcref{lem:chase is a bunch of chains}, form several chains leading to elements of $\inst$. Each of the $\pSuccz$-chains only grows until the distance to the $\pEnd$-predicate is bigger than the bound on $\encSet$ (here 2).}
 \label{fig:fig6-structure-chase}
\end{figure*}
}

Note that the above echoes Lemma~\ref{lem:chase is a chain}, except that the end point can now be any element from the original instance $\inst$.
Notice that using the semi-oblivious chase also prevents branching: each element from $\adom(\inst)$ can only be the end point of one such chain of fresh elements.
The result of the chase is depicted in \zcref{fig:fig6-structure-chase}. 

We now clarify that, as the chase progresses, new facts either stem from the flooding mechanism or involve one of the freshly introduced starting points of a chain. 

\begin{lemmarep}\label{lem:magic}
	Let $n \geq 0$ and $\apred(t_1, \dots, t_m)$ an atom.
	If $\apred(t_1, \dots, t_m) \in \inst_{n+1} \setminus \inst_n$ then one of the following holds:
	\begin{enumerate}[leftmargin=1.85em]
		\item $t_1 \in \adom(\inst_{n+1}) \setminus \adom(\inst_n)$ and there are $k_2$,\dots, $k_m \geq 0$ s.t.\ $\pSucc^{k_2}(t_1, t_2), \dots, \pSucc^{k_m}(t_1, t_m) \in \inst_{n+1}$, where  $\pSucc^k(x,y)$ denotes an $\pSucc$-path of size $k$ from $x$ to $y$.
		\item for every $1 \leq i \leq m$, we have $\pFlood(t_i) \in \inst_{n+1}$.
	\end{enumerate}
\end{lemmarep}

\begin{proofsketch}
	Apart from Rules~\ref{rule:flood propagate} and \ref{rule:flood generate} that relate to the flooding, notice that the first term occurring in an atom of some head also appears as the first term in an atom from the body.
	The rest follows by induction.
\end{proofsketch}

\begin{proof}
	Let $n \geq 0$.
	We prove that the derivation leading from $\inst_n$ to $\inst_{n+1}$ only produces atoms that respect the claimed condition. 
	We proceed by induction on this derivation.
	If the derivation is empty, that is $\inst_{n+1} = \inst_n$, then we are done as $\inst_{n+1} \setminus \inst_n = \emptyset$.
	Otherwise, let us assume that the claimed condition holds for all atoms introduced so far, and let $R$ be the next rule to be applied.
	Assume $R$ introduces an atom $\apred(t_1, \dots, t_m) \in \inst_{n+1} \setminus \inst_n$.
	We treat all nine possible cases for $R$.
	\begin{itemize}[leftmargin=1cm]
		\item[\ref{rule:successor}.]
		We have $\apred(t_1, \dots, t_m) = \pSucc(t_1, t_2)$ and $\pSuccz(t_1, t_2)$ has already been derived.
		If $\pSucc(t_1, t_2) \in \inst_{n+1} \setminus \inst_n$, we can apply the induction hypothesis which immediately concludes.
		Otherwise $\pSucc(t_1, t_2) \in \inst_{n}$, but $\inst_n$ is closed under Datalog rules, thus under Rule~\ref{rule:successor} and thus $\pSuccz(t_1, t_2) \in \inst_n$, contradiction.
		
		\item[\ref{rule:count}.]
		We have $\apred(t_1, \dots, t_m) = \pR_{i+1}(t_1, t_2)$ and there exists some $t$ such that $\pR_i(t_1, t)$ and $\pSucc(t, t_2)$ have already been derived.
		Notice that, due to $\pSucc(t, t_2)$ and $\inst_{n+1}$ being close under Datalog rules, we have $\pFlood(t_2) \in \inst_{n+1}$ due to Rule~\ref{rule:flood propagate}.
		If $\pR_i(t_1, t) \notin \inst_n$, then we can apply the induction hypothesis.
		In case $t_1 \notin \adom(\inst_{n})$, we have $\pSucc^k(t_1, t) \in \inst_{n+1}$, thus using $\pSucc(t, t_2)$ we obtain $\pSucc^{k+1}(t_1, t_2) \in \inst_{n+1}$ as desired.
		In case $\pFlood(t_1) \in \pFlood(\inst_{n+1})$, recall that also $\pFlood(t_2) \in \inst_{n+1}$ thus we are done.
		If $\pR_i(t_1, t) \in \inst_n$, note that $\pSucc(t, t_2) \in \inst_n$ as well using Lemma~\ref{lem:chase is a bunch of chains}.
		Therefore Rule~\ref{rule:count} was already applicable on $\inst_n$ which is closed under Datalog rules, so $\pR_{i+1}(t_1, t_2) \in \inst_n$, contradiction.
		
		\item[\ref{rule:mult}.]
		The argument is essentially the same as for $R = \ref{rule:count}$.
		We have $\apred(t_1, \dots, t_m) = \pT_{i}(t_1, t_2, t_3)$ and there exists some terms $s, s'$ such that $\pT_{i}(t_1, s, s')$, $\pSucc^{q_i}(s, t_2)$ and $\pSucc^{r_i}(s', t_3)$ have already been derived.
		Notice again that $\pFlood(t_2), \pFlood(t_3) \in \inst_{n+1}$.
		If we can apply the induction hypothesis on atom $\pT_{i}(t_1, s, s')$, we can thus conclude again by composing the $\pSucc$-paths or via $\pFlood(t_1) \in \inst_{n+1}$.
		Otherwise we have $\pT_{i}(t_1, s, s') \in \inst_{n}$ and thus, using Lemma~\ref{lem:chase is a bunch of chains} twice we get $\pSucc^{q_i}(s, t_2), \pSucc^{r_i}(s', t_3) \in \inst_n$.
		We could thus derive $ \pT_{i}(t_1, t_2, t_3)$ in $\inst_n$, contradiction.
		
		\item[\ref{rule:flood propagate}.]
		We have $\apred(t_1, \dots, t_m) = \pFlood(t_1)$ thus $\pFlood(t_1) \in \inst_{n+1}$ and we are done.
		
		\item[\ref{rule:successor end}.]
		We have $\apred(t_1, \dots, t_m) = \pSucc(t_1, t_1)$ and $\pEnd(t_1)$ has already been derived.
		Since we cannot produce $\pEnd$ using rules, we have $\pEnd(t_1) \in \inst$, thus $\pSucc(t_1, t_1) \in \inst_0$ contradicting $\pSucc(t_1, t_1) \notin \inst_n$.
		
		\item[\ref{rule:G init}.]
		We have $\apred(t_1, \dots, t_m) = \pG(t_1, t_2)$ and there exists some $t$ such that $\pSucc(t_1, t)$ and $\pSucc(t, t_2)$ have already been derived.
		If $\pSucc(t_1, t) \notin \inst_n$ then $t_1 \notin \adom(\inst_{n})$ using Lemma~\ref{lem:chase is a bunch of chains} and clearly $\pSucc^2(t_1, t_2) \in \inst_{n+1}$ as desired.
		Otherwise $\pSucc(t_1, t) \in \inst_n$ and therefore, using Lemma~\ref{lem:chase is a bunch of chains} we have $\pSucc(t, t_2) \in \inst_n$ as well.
		Thus Rule~\ref{rule:G init} is already applicable in $\inst_n$, which is closed under Datalog rules thus $\pG(t_1, t_2) \in \inst_n$, contradiction.
		
		\item[\ref{rule:G step}.]
		We have $\apred(t_1, \dots, t_m) = \pG(t_1, t_2)$ and there exists some term $t$ such that $\pG(t_1, t)$, $\pR_i(t_1, t)$ and $\pT_i(t_1, t, t_2)$ have already been derived.
		If $\pT_i(t_1, t, t_2) \notin \inst_n$, then we can apply the induction hypothesis to it which immediately yields the desired conclusion.
		Otherwise we have $\pT_i(t_1, t, t_2) \in \inst_n$.
		If both $\pG(t_1, t), \pR_i(t_1, t) \in \inst_n$ then Rule~\ref{rule:G step} was already applicable in $\inst_n$, contradiction.
		Otherwise we can apply the induction hypothesis to one of these, which gives, in both cases $\pFlood(t_1), \pFlood(t) \in \inst_{n+1}$ since the case of $t_1 \notin \adom(\inst_{n+1})$ would be contradicting $\pT_i(t_1, t, t_2) \in \inst_n$.
		We now consider the introduction of $\pT_i(t_1, t, t_2) \in \inst_{n}$.
		\begin{itemize}
			\item 
			If it has been introduced by the application of Rule~\ref{rule:mult}, then notice that $t_3$ has a necessary $\pSucc$-predecessor thus $\pFlood(t_3) \in \inst_{n+1}$ and we are done.
			\item
			If it has been introduced by the application of Rule~\ref{rule:flood generate}, then notice that the body guarantees $\pFlood(t_3) \in \inst_{n+1}$ and we are done.
			\item
			If it has been introduced by the application of Rule~\ref{rule:existential}, then notice that $t_1 = t = t_3$ and since $\pFlood(t_1) \in \inst_{n+1}$ we are done.
			\item
			Otherwise $\pT_i(t_1, t, t_2) \in \inst$, in which case Corollary~\ref{coro:flood} yields $\pFlood(t_1), \pFlood(t) \in \inst_{0}$ since $\pFlood(t_1), \pFlood(t) \in \inst_{n+1}$ Therefore, via Rule~\ref{rule:flood generate} and recalling that $\inst_0$ is closed under Datalog rules, we have $\pG(t_1, t), \pR_i(t_1, t) \in \inst_0 \subseteq \inst_n$.
			We reached a contradiction.
		\end{itemize}

		\item[\ref{rule:flood generate}.] 
		The body of \ref{rule:flood generate} guarantees that $\pFlood(t_i)$ has already been derived for all $t_i$'s and thus that $\pFlood(t_i) \in \inst_{n+1}$ for all $t_i$'s as desired.
		
		\item[\ref{rule:existential}.]
		We have $\apred(t_1, \dots, t_m) \in \{ \pSuccz(t_1 , t_2) ,\pR_0(t_1,t_1) \} \cup \{ \pT_j(t_1,t_1,t_1) \mid 1 \leq j \leq p-1 \}$ and $t_1$ is freshly introduced, thus $t_1 \notin \adom(\inst_{n})$.
		In cases $\pR_0(t_1,t_1)$ and $\pT_j(t_1,t_1,t_1)$, the claim is trivially true as $\pSucc^0(t_1, t_2) \in \inst_{n+1}$ (that is $t_1 = t_2$).
		For the case $\apred(t_1, \dots, t_m) = \pSuccz(t_1 , t_2)$, note that it makes the rule \ref{rule:successor} applicable.
		Since we later saturate by Datalog rules in the derivation from $\inst_n$ to $\inst_{n+1}$, this guarantees that $\pSucc(t_1, t_2)$ will further be derived due to Rule~\ref{rule:successor}.
		Hence $\pSucc(t_1, t_2) \in \inst_{n+1}$ as desired.
	\end{itemize}
\end{proof}

The above guarantees that any element $a_i$, as described in Lemma~\ref{lem:chase is a bunch of chains}, appears \emph{exactly} at depth $i$ in the chase:

\begin{lemmarep}\label{lem:depth is exact}
	Let $a_i$ be a term from $\adom(\inst_n)$ as described in Lemma~\ref{lem:chase is a bunch of chains}, thus with $i \leq n$. 
	We have $a_i \in \adom(\inst_i)$.
\end{lemmarep}

\begin{proofsketch}
	We proceed by induction on $i$.
	For $a_0$, the claim is trivial as $a_0 \in \adom(\inst)$.
	For $a_{i +1}$, we know it is introduced by applying Rule~\ref{rule:existential} on $a_{i}$ (mapping $y$ to $a_i$), thus relying on an atom $\pG(a_i, e)$ to be satisfied for some $\pEnd(e) \in \inst$.
	Using Lemma~\ref{lem:magic}, one can establish that $\pG(a_i, e) \in \inst_{i}$, thus $a_{i+1}$ is introduced at step $i+1$.
\end{proofsketch}

\begin{proof}
	We proceed by induction on $n$.
	
	For $n = 0$, $\adom(\inst_0) = \adom(\inst)$, so since by definition $a_0 = a$ for all $a \in \adom(\inst)$, the claim holds.
	
	Let us assume the claim holds at some step $n \geq 0$ and consider a term $a_i \in \adom(\inst_{n+1})$.
	If $a_i \in \adom(\inst_n)$ then we are done by induction hypothesis.
	Otherwise $a_i \in \adom(\inst_{n+1}) \setminus \adom(\inst_n)$, thus $a_i$ has been freshly introduced by some application of \ref{rule:existential} mapping $y$ on some term $t$ of $\inst_{n}$ and $z$ on a term $e$ such that $\pG(t, e), \pEnd(e) \in \inst_n$ (recall that, in the derivation from $\inst_n$ to $\inst_{n+1}$, we apply first the existential rule \ref{rule:existential}, which does not produce any $\pG$- nor $\pEnd$-atoms, which is why we know $\pG(t, e), \pEnd(e) \in \inst_n$).
	Using Lemma~\ref{lem:chase is a bunch of chains}, we deduce that $i \geq 1$ and that $t = a_{i - 1}$.
	Since $a_{i - 1} \in \adom(\inst_n)$, by induction hypothesis we obtain $a_{i-1} \in \adom(\inst_{i-1})$.
	We now want to prove that $\pG(a_{i - 1}, e), \pEnd(e) \in \inst_{i-1}$.
	This would conclude as it proves that \ref{rule:existential} is applicable on $\inst_{i-1}$ and therefore that $a_i \in \adom(\inst_i)$ as desired.
	
	Since no rule can ever produce $\pEnd$-atoms, we have $\pEnd(e) \in \inst$, and in particular $\pEnd(e) \in \inst_{i-1}$.
	It remains to prove that $\pG(a_{i - 1}, e) \in \inst_{i-1}$.
	Let us now set $j$ the smallest integer such that $\pG(a_{i-1}, e) \in \inst_j$ (recall we know $\pG(a_{i-1}, e) \in \inst_{n}$, so this is well-defined).
	If $j = 0$, then $\pG(a_{i-1}, e) \in \inst_0$, thus $\pG(a_{i-1}, e) \in \inst_{i-1}$ and we are done.
	Otherwise $j \geq 1$ and we apply Lemma~\ref{lem:magic} to the fact $\pG(a_{i-1}, e) \in \inst_j \setminus \inst_{j-1}$.
	If we are in the case $a_{i-1} \in \adom(\inst_j) \setminus \adom(\inst_{j-1})$, then it must be that $j = i-1$ since $a_{i-1} \in \adom(\inst_{i-1})$, thus $\pG(a_{i-1}, e) \in \inst_{i-1}$ and we are done.
	Otherwise we have that $\pFlood(a_{i-1}) \in \inst_j$.
	Due to our rules, this is only possible if $\pFlood(a_{i-1}) \in \inst$ or if $a_{i-1}$ has an $\pSucc$-predecessor in $\inst_j$.
	In the first case, we have $\pFlood(a_{i-1}) \in \inst$ and recall that $\pEnd(e) \in \inst$.
	Thus \ref{rule:successor end} and \ref{rule:flood propagate} apply on $e$ already in $\inst$, making \ref{rule:flood generate} applicable with $x = a_{i-1}$, and $y = z = e$, thus producing atom $\pG(a_{i-1}, e)$ already in $\inst_0$, and we are done.
	In the latter case, $a_{i-1}$ has an $\pSucc$-predecessor in $\inst_j$, which cannot be $a_i$ yet, so, according to Lemma~\ref{lem:chase is a bunch of chains}, we have $a_{i-1}$ has an $\pSucc$-predecessor in $\inst_0$.
	Here again, \ref{rule:successor end} applies on $e$ in $\inst$ and \ref{rule:flood propagate} apply on both $a_{i-1}$ and $e$ already in $\inst_0$, making \ref{rule:flood generate} applicable with $x = a_{i-1}$, and $y = z = e$, thus producing atom $\pG(a_{i-1}, e)$ already in $\inst_0$, and we are done.
\end{proof}

\begin{toappendix}
	We now state a direct corollary of Lemmas~\ref{lem:chase is a bunch of chains} and \ref{lem:depth is exact}, simply stating that if $t$ is the $\pSucc$-predecessor of $t'$, and that $t'$ is a term of $\inst_n$, then $t$ is a term of $\inst_{n+1}$.
	In the upcoming proof of Lemma~\ref{lem:central part stabilizes}, it will spare us several case distinctions between terms from $\inst$ and fresh terms introduced along a chain.
	\begin{corollary}\label{coro:tied depths}
		For every $n \geq 0$, if $\pSucc(t, t') \in \inst_n$ and $t' \in \adom(\inst_i)$ for some $0 \leq i \leq n$, then $t \in \adom(\inst_{i+1})$.
	\end{corollary}
	\begin{proof}
		We first rely on the description of the $\pSucc$-atoms in $\inst_n$.
		If both $t, t' \in \adom(\inst)$, the claim is clear.
		Otherwise $t = a_{j+1}$ and $t' = a_j$ for some $0 \leq j \leq n$ and $a \in \adom(\inst)$.
		We have $a_j \in \adom(\inst_i)$ by assumption, thus $j \leq i$ by Lemma~\ref{lem:chase is a bunch of chains} applied to $\inst_i$.
		Therefore $j + 1 \leq i + 1$ and Lemma~\ref{lem:depth is exact} now gives $t = a_{j + 1} \in \adom(\inst_{i+1})$ as desired.
	\end{proof}
\end{toappendix}

To proceed with the construction of homomorphism $\ahom' : \query \rightarrow \inst_\stopthecount$ as previously sketched. We thus decompose the homomorphism $\ahom : \query \rightarrow \inst_N$ to identify which parts of $\query$ are mapped along $\pSuccz$-$\pSucc$-chains.
To do so, we define the $\set{\pSuccz, \pSucc}$-graph of an instance $\mathcal{J}$: its nodes are elements of $\adom(\mathcal{J})$ and there is an edge from $a$ to $b$ if and only if $\pSuccz(a,b)$ or $\pSucc(a,b)$ belongs to $\mathcal{J}$. 
We denote this graph by $\SGraph^{\mathcal{J}}$. 
We then consider weakly connected components (WCC) of this graph. Two terms $a$ and $b$ are in the same WCC if there is an undirected path from $a$ to $b$ in $\SGraph^{\mathcal{J}}$.

We now come back to defining $\ahom'$. Let $V_1, \dots, V_\ell$ be the WCCs of $\SGraph^{\query}$.
We apply Lemma~\ref{lem:chase is a bunch of chains} to $\inst_N$ and obtain the terms as described in its statement.
For each $1 \leq l \leq \ell$, let $d_l$ be the minimum $d$ such that $h(v) = a_d$ for some $v \in V_l$ and $a \in \adom(\inst)$.
We define the mapping $\ahom' : \Vars(\query) \rightarrow \adom(\inst_\stopthecount)$ by `shifting' each WCC forward along the chain it lies on, unless it is already close to $\inst$:
\begin{align*}
	\ahom'(v) \coloneqq 
	\left\lbrace
	\begin{array}{ll}
		h(v) & \text{ if } d_l \leq \size{\query} \text{ where } v \in V_l
		\\
		a_{i - d_l} & \text{ otherwise, if } h(v) = a_i \text{ and } v \in V_l
	\end{array}
	\right.
\end{align*}
We first verify that $\ahom'$ indeed maps elements within $\adom(\inst_\stopthecount)$.
Consider a variable $v$ of $\query$ and let $V_l$ be its WCC in $\SGraph^{\query}$.
Let $v_{min}$ be the variable of $V_l$ that realizes $d_l$, for which $h(v_{min}) = a_{d_l}$ for some $a \in \adom(\inst)$.
By definition of $V_l$, there exists an undirected (simple) path in $\SGraph^{\query}$ from $v$ to $v_{min}$, whose length is at most $\size{\query}$.
From $\ahom$ being a homomorphism, the atoms that form this path are mapped via $\ahom$ on atoms described by Lemma~\ref{lem:chase is a bunch of chains}.
It is thus clear that $h(v)$ is some $a_i$ for $a \in \adom(\inst)$ and $0 \leq i \leq d_l + \size{\query}$.
Now, if $d_l \leq \size{q}$, then $\ahom'(v) = \ahom(v) = a_i$ and we have $i \leq 2\size{\query}$.
Otherwise $\ahom'(v) = a_{i - d_l}$ and $i - d_l \leq \size{\query} \leq 2\size{\query}$.
In both cases, Lemma~\ref{lem:depth is exact} warrants $\ahom'(v) \in \adom(\inst_\stopthecount)$.

We now claim that $\ahom'$ is a homomorphism from $\query$ to $\inst_\stopthecount$. 
Let us denote the restriction of $\inst_N$ to $\adom(\inst_n)$ for some $n \leq N$ by $\inst_N\vert_{n}$.
We actually prove that $\ahom'$ is a homomorphism from $\query$ to $\inst_N\vert_{\stopthecount-1}$.
This is sufficient as all atoms from $\inst_N$ that only involve terms from $\adom(\inst_{M-1})$ are already in $\inst_M$.
This is stated in the following lemma, whose proof is by induction and heavily relies on the previous lemmas:
\begin{lemmarep}\label{lem:central part stabilizes}
	For every $n \geq M$, $\inst_{n}\vert_{M - 1} = \inst_{M}\vert_{M-1}$.
	In particular, for $n = N$ we have $\inst_N\vert_{\stopthecount-1} = \inst_\stopthecount\vert_{\stopthecount-1}$.
\end{lemmarep}

\begin{proof}
	Notice that $\inst_{M-1}$ is well-defined since $M \geq 1$.
	
	We proceed by induction on $n$.
	
	For $n = M$, the claim is trivial.
	
	Let us now assume that the claim holds until some $n \geq M$, that is we have $\inst_{n}\vert_{M - 1} = \inst_{M}\vert_{M-1}$.
	Notice that $\inst_{n + 1}\vert_{M - 1} \supseteq \inst_{M}\vert_{M-1}$ is by definition; it remains to prove $\inst_{n + 1}\vert_{M - 1} \subseteq \inst_{M}\vert_{M-1}$.
	We proceed by induction on the derivation $D$ that leads from $\inst_{n}\vert_{M - 1}$ to $\inst_{n + 1}\vert_{M - 1}$.
	If $D$ is empty, we have $\inst_{n + 1}\vert_{M - 1} = \inst_{n}\vert_{M - 1}$ and the claim holds by the induction hypothesis on $n$.
	Otherwise, let us assume that the claim holds for all atoms of $\inst_{n + 1}\vert_{M - 1}$ introduced so far by $D$, and let $R$ be the next rule to be applied.
	Assume the application of $R$ leads to the introduction of an atom $\apred(t_0, \dots, t_m) \in \inst_{n + 1}$ with $t_i \in \adom(\inst_{M-1})$ for every $1 \leq i \leq m$.
	We treat all nine possible cases for $R$.
	\begin{itemize}[leftmargin=1cm]
		\item[\ref{rule:successor}.]
		We have $\apred(t_1, \dots, t_m) = \pSucc(t_1, t_2)$ and $\pSuccz(t_1, t_2)$ has already been derived.
		In particular, $\pSuccz(t_1, t_2) \in \inst_{n+1}$.
		Since $t_1, t_2 \in \adom(\inst_{M-1})$ by assumption, the induction hypothesis on $D$ yields $\pSuccz(t_1, t_2) \in \inst_{M}\vert_{M - 1}$.
		Using Lemma~\ref{lem:chase is a bunch of chains}, we thus obtain $\pSuccz(t_1, t_2) \in \inst_{M-1}$.
		Therefore $\ref{rule:successor}$ is already applicable on $\inst_{M-1}$ and thus $\pSucc(t_1, t_2) \in \inst_M$ as desired.
		
		\item[\ref{rule:count}.]
		We have $\apred(t_1, \dots, t_m) = \pR_{i+1}(t_1, t_2)$ and there exists some $t$ such that $\pR_i(t_1, t)$ and $\pSucc(t, t_2)$ have already been derived.
		In particular, $\pR_i(t_1, t), \pSucc(t, t_2) \in \inst_{n+1}$.
		If $t \in \adom(\inst_{M - 1})$ then we can conclude by applying the induction hypothesis for $D$ on both $\pR_i(t_1, t)$ and $\pSucc(t, t_2)$, which are thus already in $\inst_M$, which is closed under our Datalog rules, in particular under \ref{rule:count}, thus $\pR_{i+1}(t_1, t_2) \in \inst_{M}$.
		We now focus on the case of $t \notin \adom(\inst_{M - 1})$ and apply Corollary~\ref{coro:tied depths}, recalling that $\pSucc(t, t_2)$ and $t_2 \in \adom(\inst_{M-1})$ to obtain $t \in \adom(\inst_M)$.
		In particular, $\pSucc(t, t_2) \in \inst_{M}$ by Lemma~\ref{lem:chase is a bunch of chains}.
		We now distinguish two cases depending on whether $\pR_{i}(t_1, t) \in \inst_M$ or not.
		If $\pR_{i}(t_1, t) \in \inst_M$, then closure under Datalog rules of $\inst_M$ concludes as \ref{rule:count} is already applicable in $\inst_M$.
		Otherwise $\pR_{i}(t_1, t) \notin \inst_M$ and we can apply Lemma~\ref{lem:magic} to obtain either $t_1 \in \adom(\inst_{M}) \setminus \adom(\inst_{M-1})$ or $\pFlood(t_1) \in \inst_M$.
		The former is a contradiction with our assumption that $t_1 \in \adom(\inst_{M-1})$.
		The latter means that \ref{rule:flood generate} triggers already in $\inst_M$ for $x = t_1$ and $y = z = t_2$ (since $\pFlood(t_2) \in \inst_M$ follows from $\pSucc(t, t_2) \in \inst_{M}$), thus introducing atom $\pR_{i+1}(t_1, t_2) \in \inst_M$ as desired.
		
		\item[\ref{rule:mult}.]
		The argument is essentially the same as for $R = \ref{rule:count}$.
		We have $\apred(t_1, \dots, t_m) = \pT_{i}(t_1, t_2, t_3)$ and there exists some terms $s, s'$ such that $\pT_{i}(t_1, s, s')$, $\pSucc^{q_i}(s, t_2)$ and $\pSucc^{r_i}(s', t_3)$ have already been derived.
		In particular $\pT_{i}(t_1, s, s'), \pSucc^{q_i}(s, t_2), \pSucc^{r_i}(s', t_3) \in \inst_{n+1}$.
		If $s, s' \in \adom(\inst_{M-1})$ then we can again conclude using the induction hypothesis, proving that Rule~\ref{rule:mult} was already applicable in $\inst_{M}$.
		Otherwise, using $\pSucc^{q_i}(s, t_2), \pSucc^{r_i}(s', t_3) \in \inst_{n+1}$, we have two terms $u, u'$ such that $\pSucc(u, t_2), \pSucc(u', t_3) \in \inst_{n+1}$.
		We use Corollary~\ref{coro:tied depths} to obtain $u, u' \in \adom(\inst_M)$, recalling that $t_2, t_3 \in \adom(\inst_{M-1})$. 
		We now distinguish two cases depending on whether $\pT_{i}(t_1, s, s') \in \inst_M$ or not.
		If $\pT_{i}(t_1, s, s') \in \inst_M$, then closure under Datalog rules of $\inst_M$ concludes as \ref{rule:mult} is already applicable in $\inst_M$.
		Otherwise $\pT_{i}(t_1, s, s') \notin \inst_M$ and we can apply Lemma~\ref{lem:magic} to obtain either $t_1 \in \adom(\inst_{M}) \setminus \adom(\inst_{M-1})$ or $\pFlood(t_1) \in \inst_M$.
		The former is a contradiction with our assumption that $t_1 \in \adom(\inst_{M-1})$.
		The latter means that \ref{rule:flood generate} triggers already in $\inst_M$ for $x = t_1$ and $y = t_2$ and $z = t_3$, thus introducing atom $\pT_{i+1}(t_1, t_2, t_3) \in \inst_M$ as desired.

		\item[\ref{rule:flood propagate}.]
		We have $\apred(t_1, \dots, t_m) = \pFlood(t_1)$ and there exists a term $t$ such that $\pSucc(t, t_1)$ has already been derived.
		From $t_1 \in \adom(\inst_{M-1})$, applying Corollary~\ref{coro:tied depths} yields $t \in \adom(\inst_M)$.
		Then closure under Datalog rules of $\inst_M$ guarantees $\pFlood(t_1) \in \inst_M$ as \ref{rule:count} is already applicable in $\inst_M$ and we are done.
		
		\item[\ref{rule:successor end}.]
		We have $\apred(t_1, \dots, t_m) = \pSucc(t_1, t_1)$ and $\pEnd(t_1)$.
		Since $\pEnd(t_1)$ cannot be obtained by any rule, we have $\pEnd(t_1) \in \inst$, thus $\pEnd(t_1) \in \inst_M$ as desired. 
		
		\item[\ref{rule:G init}.]
		We have $\apred(t_1, \dots, t_m) = \pG(t_1, t_2)$ and there exists a term $t$ such that $\pSucc(t_1, t)$ and $\pSucc(t, t_2)$ have already been derived.
		Since we have $t_1 \in \adom(\inst_{M-1})$ and $\pSucc(t_1, t) \in \inst_{n+1}$, using Lemma~\ref{lem:chase is a bunch of chains} we get $t \in \adom(\inst_{M-1})$.
		We can thus use the induction hypothesis on both $\pSucc(t_1, t)$ and $\pSucc(t, t_2)$, that thus both belong to $\inst_M$.
		Since $\inst_M$ is closed under our Datalog rules, in particular under \ref{rule:G init}, we obtain $\pG(t_1, t_2) \in \inst_{M}$ as desired.
		
		\item[\ref{rule:G step}.]
		We have $\apred(t_1, \dots, t_m) = \pG(t_1, t_2)$ and we have atoms $\pG(t_1, t)$, $\pR_{i}(t_1, t)$ and $\pT_i(t_1, t, t_2)$ that have already been derived.
		We consider the smallest $j \leq n+1$ such that $\pT_i(t_1, t, t_2) \in \inst_{j+1} \setminus \inst_j$ and apply Lemma~\ref{lem:magic} on that fact.
		If we have some $k \geq 0$ such that $\pSucc^k(t_1, t) \in \inst_j$, then, joint with Lemma~\ref{lem:chase is a bunch of chains} it guarantees $t \in \adom(\inst_{M-1})$.
		We can thus apply the induction hypothesis on all three atoms $\pG(t_1, t)$, $\pR_{i}(t_1, t)$ and $\pT_i(t_1, t, t_2)$ and conclude with closure of $\inst_M$ under our Datalog rules.
		Otherwise we have $\pFlood(t_1), \pFlood(t_2) \in \inst_j \subseteq \inst_{n+1}$.
		Since $t_1, t_2 \notin \adom(\inst_{n+1}) \setminus \adom(\inst_n)$ as $n +1 > M$ and $t_1, t_2 \in \adom(\inst_{M-1})$, we have $\pFlood(t_1), \pFlood(t_2) \in \inst_n$.
		Therefore, by induction hypothesis on $n$, we have $\pFlood(t_1), \pFlood(t_2) \in \inst_{M}$.
		We use closure of $\inst_{M}$ under Datalog rules to derive 
		
		\item[\ref{rule:flood generate}.] 
		The body of \ref{rule:flood generate} guarantees that $\pFlood(t_1), \dots, \pFlood(t_m)$ have all already been derived.
		Thus by induction hypothesis we have $\pFlood(t_1), \dots, \pFlood(t_m) \in \inst_M$.
		Closure of $\inst_M$ under Datalog rules then concludes.
		
		\item[\ref{rule:existential}.]
		All atoms in the head of \ref{rule:existential} features the freshly introduced element.
		But since $n > M-1$ and that all $t_i$'s belong to $\adom(\inst_{M-1})$, the atom $\apred(t_1, \dots, t_m)$ cannot possibly have been introduced by \ref{rule:existential}.
		
	\end{itemize}
\end{proof}
We finally prove that $h'$ is a homomorphism from $\query$ to $\inst_N\vert_{\stopthecount-1}$.
Since $\pEnd$-atoms cannot be introduced during the chase, any $\pEnd$-atom $\pEnd(t)$ is mapped through $\ahom$ to an atom of $\inst$, and thus $\ahom'(t)=\ahom(t)$, so $\pEnd(\ahom'(t))\in\inst_N\vert_{\stopthecount-1}$.
The case of $\pSuccz$- and $\pSucc$-atoms is also easy: consider an atom $\apred(t, t') \in \query$ for $\apred \in \{ \pSuccz, \pSucc \}$. 
By definition, $t$ and $t'$ belong to a same WCC $V_l$.
If $\ahom$ and $\ahom'$ coincide on $V_l$ (that is the case if $d_l \leq \size{\query}$), then $\apred(\ahom'(t), \ahom'(t'))$ immediately belongs to $\inst_\stopthecount$ via Lemma~\ref{lem:central part stabilizes}.
Otherwise, $\ahom(t) = a_i$ and $\ahom(t') = a_j$ for some $a \in \adom(\inst)$ and $i, j$ such that $i=j+1$ or $i=j-1$ by Lemma~\ref{lem:chase is a bunch of chains}. Thus, $\ahom'(t) = a_{i - d_l}$ and $\ahom'(t') = a_{j - d_l}$, and $i - d_l$ and $j - d_l$ satisfy the same relation and are smaller than $\stopthecount$, so $\apred(\ahom'(t), \ahom'(t'))$ belongs to $\inst_\stopthecount$.

It remains to treat the case of atoms based upon other predicates.
Consider an atom $\apred(t_1, \dots, t_m) \in \query$ for some $\apred \notin \{ \pSuccz, \pSucc \}$ of arity $m$.
Since $\ahom$ is a homomorphism, we have $\apred(\ahom(t_1), \dots, \ahom(t_m)) \in \inst_N$.
If all $\ahom(t_i) \in \adom(\inst)$, then all $\ahom'(t_i) = \ahom(t_i)$ and we are done.
Otherwise, we can apply Lemma~\ref{lem:magic}.
In the case where $\ahom(t_i)$ are all flooded, then so are the $\ahom'(t_i)$ and Rule~\ref{rule:flood generate} concludes.
Otherwise, all $\ahom(t_i)$ belong to the same WCC $V_l$.
Again, if $d_l \leq \size{\query}$, then $\ahom'$ and $\ahom$ coincide on $V_l$ and we are done.
Otherwise, all $t_i$'s map to some $a_{k_i}$ for $a\in\adom(\inst)$ and $k_i<m_a$, since $a_{k_i+d_l}\in\adom(\inst_N)$. Thus, by Rule~\ref{rule:flood propagate}, $\ahom'(t_i)$ is flooded, and thus again Rule~\ref{rule:flood generate} concludes.

\section{Conclusion}
\label{sec:conclusion}

We have shown that decidable BCQ entailment does not imply the decidability of \bts membership or chase termination for concrete classes of existential rules. This explains the necessity of diverse, class-specific techniques for proving termination, as BCQ entailment offers no systematic tool for this purpose.
Future work includes proving that entailment of recursive C2RPQs is also decidable for the class presented here. Conversely, it remains an open problem to identify a homomorphism-closed query language whose decidable entailment is sufficient to guarantee the decidability of chase termination.

\section*{Acknowledgments}

This work has been partially supported by ANR grant EXPAND (ANR-25-CE23-1215).
The authors would like to thank David Carral who hinted this topic and provided useful insights.
He still has not played \textit{Celeste} though.

\bibliographystyle{named}
\bibliography{ijcai26}

\clearpage
\appendix

\end{document}